\documentclass[a4paper,11pt]{article}

\usepackage[margin=1in]{geometry}
\usepackage{amsmath,amsthm,amssymb}
\allowdisplaybreaks

\usepackage{hyperref}
\hypersetup{
	colorlinks=true,
	citecolor = blue!80!black,
	linkcolor = red!80!black
}

\newcommand{\email}[1]{\href{mailto:#1}{\nolinkurl{#1}}}
\usepackage[square]{natbib}
\usepackage{xcolor}
\usepackage{paralist}
\usepackage{tikz}
\usepackage{booktabs}
\usepackage{pgfplots}

\usepackage{titlesec}
\titlespacing{\paragraph}{%
	0pt}{
	0.5\baselineskip}{
	1em}
\titlespacing*{\section}
{0pt}{2ex plus 1ex minus .2ex}{1ex plus .2ex}
\titlespacing*{\subsection}
{0pt}{1ex plus 1ex minus .2ex}{0.5ex plus .2ex}

\newtheoremstyle{theoremc}
  {4pt}
  {4pt}
  {\itshape}
  {0pt}
  {\bfseries}
  {.}
  { }
  {\thmname{#1}\thmnumber{ #2}\thmnote{ (#3)}}
\theoremstyle{theoremc}

\newtheorem{theorem}{Theorem}[section]
\newtheorem{lemma}[theorem]{Lemma}
\newtheorem{definition}{Definition}
\newtheorem{proposition}[theorem]{Proposition}

\newtheorem{corollary}[theorem]{Corollary}

\theoremstyle{remark}

\newcommand{\ie}{i.e.,\xspace}
\newcommand{\eg}{e.g.,\xspace}

\def\clap#1{\hbox to 0pt{\hss#1\hss}}

\usepackage[ruled,vlined]{algorithm2e}

\renewcommand{\vec}[1]{\mathbf{#1}}

\newcommand{\vr}{\vec{r}}
\newcommand{\vs}{\vec{s}}

\newcommand{\vx}{\vec{x}}
\newcommand{\R}{\mathbb{R}}
\newcommand{\m}{\mathbb}
\newcommand{\N}{\mathbb{N}}
\newcommand{\E}[2][]{\mathbb{E}_{#1}\left[#2\right]}

\renewcommand{\Pr}[1]{\text{Pr}\left[#1\right]}

\newcommand{\eps}{\varepsilon}

\title{Unknown I.I.D.\@ Prophets: Better Bounds, Streaming\\Algorithms, and a New Impossibility}

\author{%
	\ifdefined\hideauthors
	Authors Undisclosed
	\else
	Jos\'e Correa\thanks{
	Departamento de Ingenier\'ia Industrial, Universidad de Chile, Avenue Rep\'ublica 701, Santiago, Chile. Email: \email{correa@uchile.cl}. Supported in part by ANID Chile through grant CMM-AFB 170001.} 
	\and 
	Paul D\"utting\thanks{Department of Mathematics, London School of Economics, Houghton Street, London WC2A 2AE, UK. Email: \email{p.d.duetting@lse.ac.uk}. Part of the work was done while the author was at Google Research, Z\"urich, Switzerland.}
	\and
	Felix Fischer\thanks{School of Mathematical Sciences, Queen Mary University of London, Mile End Road, London E1 4NS, UK. Email: \email{felix.fischer@qmul.ac.uk}. Supported in part by EPSRC grant~EP/T015187/1.}
	\and
	Kevin Schewior\thanks{Department Mathematik/Informatik, Abteilung Informatik, Universit\"at zu K\"oln, Weyertal 121, 50931, K\"oln, Germany. Email: \email{kschewior@gmail.com}}
	\and
	Bruno Ziliotto\thanks{Centre de Recherche en Math\'ematiques de la D\'ecision, Universit\'e Paris-Dauphine, Place du Mar\'echal de Lattre de Tassigny, 75016 Paris, France. Email: \email{ziliotto@math.cnrs.fr}.}\fi}

\date{}

\begin{document}

\maketitle

\begin{abstract}
A prophet inequality states, for some $\alpha\in[0,1]$, that the expected value achievable by a gambler who sequentially observes random variables $X_1,\dots,X_n$ and selects one of them is at least an $\alpha$ fraction of the maximum value in the sequence. 
We obtain three distinct improvements for a setting that was first studied by Correa et al.\ (EC,~2019) and is particularly relevant to modern applications in algorithmic pricing. In this setting, the random variables are i.i.d.\ from an unknown distribution and the gambler has access to an additional $\beta n$ samples for some $\beta\geq 0$. 
We first give improved lower bounds on $\alpha$ for a wide range of values of $\beta$; specifically, $\alpha\geq(1+\beta)/e$ when $\beta\leq 1/(e-1)$, which is tight, and $\alpha\geq 0.648$ when $\beta=1$, which improves on a bound of around $0.635$ due to Correa et al.\@ (SODA,~2020).
Adding to their practical appeal, specifically in the context of algorithmic pricing, we then show that the new bounds can be obtained even in a streaming model of computation and thus in situations where the use of relevant data is complicated by the sheer amount of data available.
We finally establish that the upper bound of $1/e$ for the case without samples is robust to additional information about the distribution, and applies also to sequences of i.i.d.\@ random variables whose distribution is itself drawn, according to a known distribution, from a finite set of known candidate distributions. This implies a tight prophet inequality for exchangeable sequences of random variables, answering a question of Hill and Kertz (Contemporary Mathematics,~1992), but leaves open the possibility of better guarantees when the number of candidate distributions is small, a setting we believe is of strong interest to applications.
\end{abstract}

\section{Introduction}
\label{sec:intro}

Prophet inequalities have been studied extensively in optimal stopping theory and have recently seen a surge of interest in theoretical computer science. They are less pessimistic than classic worst-case competitive analysis as a framework for the study of online algorithms, and have major applications in algorithmic mechanism design and algorithmic pricing. These applications include relatively obvious ones in the design and analysis of posted-price mechanisms, but also more indirect ones like the choice of reserve prices in online advertising auctions~(see, \eg the survey of \citet{HaLi16a}).

The basic prophet inequality problem---sometimes referred to as the \emph{single-choice} or \emph{single-item} prophet inequality---is the following: 
a \emph{gambler} observes a sequence of random variables $X_1,\dots,X_n$ which satisfy certain assumptions, for example that they are drawn independently from possibly distinct known distributions. Upon seeing the realization $X_i=x_i$ the gambler has to decide immediately and irrevocably whether to stop and receive $x_i$ as reward or to continue to the next random variable. Denoting the possibly randomized index at which the gambler stops by $\tau$, the expected reward of the gambler is $\mathbb{E}[X_\tau]$. The goal is to find a stopping rule that is competitive with the expected reward of an all-knowing \emph{prophet} who in particular knows the entire sequence and can simply choose the highest reward in the sequence. The goal thus is to determine the largest $\alpha\in [0,1]$ such that the inequality $\mathbb{E}[X_\tau] \geq \alpha \cdot \mathbb{E}[\max_i\{X_1, \dots, X_n\}]$ holds for all random variables that satisfy the assumptions.

When distributions are known, optimal stopping rules can in principle be found by backward induction. Determining the optimal competitive ratio $\alpha$ may, nevertheless, be an intricate task. Whereas a tight result of $\alpha=1/2$ for non-identical distributions has been known since the 1980s \citep{KrengelS77,KrengelS78,Samuel84}, a bound of $\alpha\approx 0.745$ for the case of identical distributions was only recently shown to be tight~\citep{HillK82,AEEH+17a,CorreaFHOV17}.

An exciting new set of questions arises if we assume that the distributions of the random variables are unknown. In what is arguably the most basic setting of this kind, the gambler faces $n$ i.i.d.~draws from an unknown distribution. Since the random variables come from the same distribution, we may hope to be able to learn from earlier draws and apply what we have learned to later ones. \citet{CorreaDFS19} showed that, perhaps surprisingly, no learning is possible and the optimal solution mirrors the optimal solution to the classic secretary algorithm to achieve a bound of $1/e$; the impossibility persists even with $o(n)$ additional samples from the distribution, whereas a considerably better lower bound can be achieved with $\beta n$ samples, for $\beta>0$, and $O(n^2)$ samples are enough to get arbitrarily close to the optimal bound of $\alpha\approx 0.745$ achievable for a known distribution. The latter can actually be achieved already with $O(n)$ samples~\citep{RubinsteinWW20}. While improved bounds have been obtained for the case with $\beta n$ samples~\citep{CorreaCES20,KaplanNR20}, significant gaps remain between lower and upper bounds across the whole range of values of $\beta$. In a setting with non-identical distributions, a single sample from each distribution is enough to match the optimal bound of $\alpha=1/2$ achievable with full knowledge of the distributions~\citep{RubinsteinWW20}.

\subsection{Our Contribution}

We improve the results of \citet{CorreaDFS19} and the follow-up work in three distinct directions. Firstly, we give a new prophet inequality for i.i.d.~random variables from an unknown distribution and access to additional samples. The inequality is tight when the number of samples is small, and otherwise improves on the state of the art. To obtain it we consider a natural class of algorithms and use tools from variational calculus to analyze the optimal algorithm from the class exactly. Secondly, we show that the new prophet inequality can be translated with arbitrarily small loss into a streaming model of computation. While streaming algorithms for optimal stopping have to our knowledge not previously been considered, they become very interesting when we deal with unknown distributions and large amounts of relevant data. The latter are commonly found in modern applications of prophet inequalities to algorithmic pricing. 
Thirdly, we give an upper bound of $1/e$ for a setting where the random variables are i.i.d.\@ from a distribution $F^j$ which is not itself known but drawn from a known distribution over a finite set of known scenarios $F^1,\dots,F^m$. The bound is established for finite $m$ and $n$, and scenarios of finite support, and shows that the impossibility in the absence of additional samples is robust to additional information regarding the distribution. It also implies a tight prophet inequality of $1/e$ for exchangeable sequences of random variables with a known joint distribution, answering a question of \citet{HillKertz92}.

\paragraph{Unknown Distribution.}

In recent work, \citet{CorreaDFS19} introduced a prophet inequality problem involving~$n$ i.i.d.\@ random variables from an unknown distribution and $k\geq 0$ additional samples. The main motivation for studying this problem arises from modern applications of prophet inequalities, specifically their use in the analysis of posted-price mechanisms and reserve pricing in advertising auctions. While it is common in these applications to model valuations as draws from an underlying distribution, it may not be reasonable to assume that the distribution is known to the auctioneer. The auctioneer may, however, choose to learn the distribution on the fly as opportunities arrive, or may possess some limited historical information in the form of additional samples.

\citet{CorreaDFS19} showed that in the absence of samples the correct value of $\alpha$ is $1/e$. They also gave parametric upper and lower bounds for the case of $\beta n$ samples when $\beta>0$, which for $k=n$ are equal to $1-1/e\approx 0.6321$ and $\ln(2)\approx 0.69$. Follow-up work has produced improved lower bounds for $k<n$~\citep{KaplanNR20} and $k=n$~\citep{CorreaCES20}, the latter being equal to approximately~$0.635$.

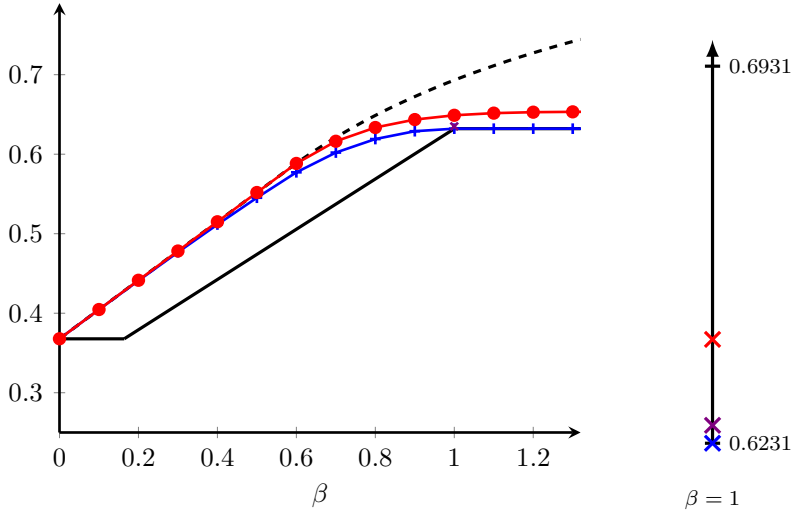
\begin{figure}[t]
\centering
\begin{tikzpicture}
\begin{axis}[xlabel=$\beta$, axis x line = left, axis y line = left, xmin = 0.0, xmax = 1.32, ymin = 0.25, ymax = 0.79, font=\small, line width = 1pt] 
\addplot[no marks,very thick,domain=0:0.58197,dashed] {(1+x)/exp(1)}; 
\addplot[no marks,very thick,domain=0.58197:1.32,dashed] {-x*ln(x/(1+x))};
\addplot[no marks,very thick,domain=0:0.164] {0.3678};
\addplot[no marks,very thick,domain=0.164:1] {0.316*x+0.316};
\addplot[no marks,very thick,domain=1:1.32] {0.6321};
\addplot[solid, blue, mark=+] coordinates {
  (1.4, 0.632120)
  (1.3, 0.632120)
  (1.2, 0.632120)
  (1.1, 0.632120)
  (1.0, 0.632120)
  (0.9, 0.628912)
  (0.8, 0.619052)
  (0.7, 0.602063)
  (0.6, 0.577208)
  (0.5, 0.545239)
  (0.4, 0.511545)
  (0.3, 0.476724)
  (0.2, 0.440991)
  (0.1, 0.404608)
  (0.0, 0.367880)
};
\addplot[solid, red, mark=*] coordinates {
  (1.4, 0.653368)
  (1.3, 0.653280)
  (1.2, 0.652853)
  (1.1, 0.651654)
  (1.0, 0.648957)
  (0.9, 0.643563)
  (0.8, 0.633580)
  (0.7, 0.616281)
  (0.6, 0.588379) 
  (0.5, 0.551819)
  (0.4, 0.515031)
  (0.3, 0.478243) 
  (0.2, 0.441455)
  (0.1, 0.404667)
  (0.0, 0.367880)
};
\addplot[solid, violet, mark=x] coordinates {
  (1.0, 0.635)
};
\end{axis}
\end{tikzpicture}
\hspace*{1cm}
\begin{tikzpicture}
\scriptsize
\draw[-latex,very thick] (1,1) node[below,yshift=-0.5cm]{$\beta = 1$} -- (1,6.35);
\draw[-,very thick] (0.9,1) -- (1.1,1) node[right] {$0.6231$};
\draw[-,very thick,blue] (0.9,0.9) -- (1.1,1.1);
\draw[-,very thick,blue] (0.9,1.1) -- (1.1,0.9);
0.2377
\draw[-,very thick,violet] (0.9,1.1377) -- (1.1,1.3377);
\draw[-,very thick,violet] (0.9,1.3377) -- (1.1,1.1377);
\draw[-,very thick,red] (0.9,2.277) -- (1.1,2.477);
\draw[-,very thick,red] (0.9,2.477) -- (1.1,2.277);
1.377
\draw[-,very thick] (0.9,6) -- (1.1,6) node[right] {$0.6931$};
\end{tikzpicture}
\caption{Visualization of the lower bound established in this paper for varying $\beta$ (solid, red), and comparison with the parametric upper bound (dashed black) and parametric lower bound (solid black) of \citet{CorreaDFS19} as well as the lower bounds of \citet{KaplanNR20} (blue) and \citet{CorreaCES20} (violet).}\label{fig:samples-bound}
\end{figure}

The guarantee of $1-1/e$ for $k=n$ samples, achieved by~\citet{CorreaDFS19} and also by \citet{KaplanNR20}, can be obtained by setting a threshold for the acceptance of each arriving value that is equal to the maximum of a uniform random subset of size $n-1$ of all samples and values seen thus far. The key ingredient in the analysis of this algorithm is a fresh-samples lemma, which states that each of the selected sets of values is distributed like $n-1$ fresh samples. The lemma implies that, conditioned on stopping, the expected value obtained is the same as that of the prophet. The competitive ratio is thus at least, and in fact equal to, the probability of stopping at all. Since conditioned on reaching a random variable we stop with probability $1/n$, this probability is $1-(1-1/n)^n \geq 1-1/e$.

We study a natural generalization of the algorithm to what we refer to as \emph{maximum-of-random-subset} (MRS) algorithms. MRS algorithms may use as the threshold in any given step a random subset \emph{of any size} of the samples and values seen so far, and are thus characterized by a function mapping each step to the size of the set used at that step.

Our analysis of MRS algorithms is based on the observation that, with the help of a fresh-samples lemma, the expected value obtained by the algorithm can be written as a linear combination of expected maxima of certain numbers of fresh draws from the distribution. On the other hand, the value obtained by the prophet, to which we are comparing ourselves, is simply the expected maximum of $n$ fresh draws. We can thus determine the worst-case ratio between the two values by writing any maximum of $\gamma n$ draws as the integral $\int_0^\infty F^{\gamma n}(x)\;\mathrm{d}x$ and then considering the integrals pointwise for each~$x$. After substituting $a=F^n(x)$, we are left with the minimization of a single rational function in a single variable. As we will see, this technique leads to a tight analysis of MRS algorithms.

To determine the best MRS algorithm using an arbitrary number of samples, we take the supremum over all functions characterizing MRS algorithms of the solution to the aforementioned single-variable optimization problem. An upper bound on the solution of this supremum-infimum (control) problem can be obtained by swapping the supremum and the infimum, solving the inner supremum using the Euler--Lagrange equation, and tackling the remaining optimization problem by multivariable calculus. A lower bound can be obtained by substituting a solution of the form obtained from the Euler--Lagrange equation into the supremum-infimum problem and again solving the remaining optimization problem by multivariable calculus. It turns out that the resulting upper and lower bounds match, which indirectly implies a minimax result. The best guarantee that can be obtained using MRS algorithms is approximately $0.653$ and can be realized with around $1.443\cdot n$ samples.

For the case in which only a bounded number of $\beta n$ samples is available, we conjecture a certain structure of the optimal MRS algorithm. Given this structure, we obtain a similar control problem as before. Using numerical methods, we obtain lower bounds for different values of $\beta$, which we conjecture to be tight up to errors in the numerical approximation.  Our lower bounds improve on the state of the art for a broad range of values of $\beta$. In particular we establish the existence of an MRS algorithm using~$n$ samples that achieves a guarantee of roughly $0.649$, which presents a considerable improvement over the previous best guarantee of $0.635$ for the exact same setting due to \citet{CorreaCES20}.

An interesting question going forward concerns the types of algorithms that could be used to close the remaining gap between the lower and upper bounds. Two natural candidates are algorithms that may skip an initial fraction of the values, effectively turning them into samples, and algorithms that may consider other (empirical) order statistics of random subsets of variables than just the maximum. We will see that the former class, which includes the classic algorithm for the secretary problem when $\beta=0$, leads to a tight guarantee of $(1+\beta)/e$ when $\beta\leq 1/(e-1)\approx 0.58$.

A summary of our results for i.i.d.~random variables from an unknown distribution and a comparison to the results from prior work can be found in \autoref{fig:samples-bound}.

\paragraph{Streaming Prophet Inequalities.}

As we have mentioned earlier, the study of prophet inequalities for i.i.d.~random variables from an unknown distribution is motivated by applications in posted pricing and in reserve pricing in advertising auctions. Such applications are often subject to an additional constraint, imposed by a vast amount of individual transactions that limits the way in which data about these transactions can be accessed and used. In particular, simply storing all past data in its entirety for later use is often impossible. 

The canonical model for studying algorithms subject to this kind of constraint is the \emph{streaming model}~(see, \eg the survey of \citet{Muthukrishnan05}). In the streaming model, an algorithm is allowed only a limited number of sequential passes over the input, often just a single one, and can only store a limited amount of information, logarithmic in the size of the input or constant. In the prophet problem, allowing only a single pass over the sequence of random variables is very natural, but the streaming model additionally limits the amount of memory available to an algorithm to $O(\log (k+n))$ or $O(1)$.

Streaming algorithms for the prophet problem have to our knowledge not previously been studied. This should not be surprising, since stopping rules typically rely on information about underlying distributions like their mean or quantiles, and distributions were until very recently assumed to be known from the outset. When distributions are not known, as in the model of \hbox{\citet{CorreaDFS19}}, any information about them must be inferred from the sequence of random variables and streaming algorithms suddenly become very desirable. This is true in particular when optimal stopping is used as a tool for modern applications in algorithmic pricing, which often involve very large amounts of relevant data.

MRS algorithms as described above are not streaming algorithms, as they may require access to arbitrary subsets of the random variables seen so far at any point. We will see, however, that any MRS algorithm can be implemented as a streaming algorithm with $\eps$ additive loss in the guarantee by storing $O_\eps(1)$ samples. Note that computing a threshold for accepting the next value for \emph{each} such value independently would essentially require remembering all of the seen values. Instead, we recycle the first computed threshold for some number of times steps and recompute it whenever it becomes too bad of an approximation for the threshold the MRS algorithm would actually set. We show that it suffices to (non-adaptively, i.e., independently of the observed values) choose $O_\eps(1)$ time points for recomputing the threshold. The computation of the thresholds is implemented by an on-the-fly construction of the corresponding random subsets which only requires to store a single sample and $O(\log n)$ additional space.

\citet{CorreaDFS19} showed that $O_\eps(n^2)$ samples are enough to get within $\eps$ of the tight bound of around $0.745$ for the case of a known distribution. Quite remarkably, the same can already be achieved with just $O_\eps(n)$ samples~\citep{RubinsteinWW20}. The algorithms underlying both of these results rely heavily on \emph{empirical quantiles}, and strong communication-complexity lower bounds for quantile estimation~\citep{GuhaM09} suggest that they cannot be implemented in the streaming model. This suggests that a gap might exist between the performance achievable respectively by streaming algorithms and algorithms that have random access to past random variables. We leave this interesting question for future work.

\paragraph{Unknown Distribution with Prior.}

The prophet problem for i.i.d.\@ random variables from an unknown distribution is subject to a strong impossibility result that matches the obvious lower bound of $1/e$, and any improvement requires a significant number of additional samples. It is natural to ask whether an improvement is possible if information about the distribution is available not in the form of samples but more directly.

We explore this question by considering a setting where a distribution $F^j$ is drawn from a set of distributions $F^1,\dots,F^m$ according to a prior distribution $\theta$, and the gambler is presented with random variables $X_1,\dots,X_n$ drawn independently from $F^j$. The gambler knows the prior distribution $\theta$ including its support, but not the distribution $F^j$ that has been drawn. This setting models situations where a decision maker is aware of the existence of a number of scenarios and has a good understanding of each scenario, but does not know which scenario is currently unfolding. For example, a merchant may be aware that there are good and bad days for selling a particular item, and may be aware how the valuation potential buyers have for the item is distributed on good days and on bad days, but may not be aware whether it is a good or a bad day.

Given the additional information about the distribution, we may hope to be able to improve on the obvious lower bound of $1/e$. We will see, however, that no such improvement is possible even when $n$, $m$, and the distributions $F^1,\dots,F^m$ are all finite. This shows that the impossibility result for unknown distributions is rather robust. It also implies tightness of a known prophet inequality of $1/e$ due to \citet{EltonKertz91} for exchangeable sequences of random variables with a known joint distribution, answering a question of \citet{HillKertz92}. Here, a sequence of random variables is called exchangeable if their joint distribution is invariant under permutations of the sequence, which is the case for the sequence of random variables for which we establish the upper bound. The upper bound does not preclude the existence of better lower bounds in settings where $n$ or $m$ is small. We believe such settings to be of significant practical interest and leave their study as an interesting problem for future work.

The proof of the upper bound uses an appropriate minimax argument to provide a reduction to the impossibility result for an unknown distribution. However, the argument itself requires a few new ideas and a modification of the key construction.

The basic idea behind our proof is to interpret the new prophet problem as a two-player zero-sum game, or equivalently a min-max problem, where the first player chooses a prior $\theta$ and the second player chooses a stopping rule. The payoff that the first player seeks to minimize and the second player seeks to maximize is the expected reward from the stopping rule minus $1/e$ times the expected maximum in the sequence. If we could reverse the order of minimization and maximization, and thus turn the problem into a max-min problem, we would be looking at a situation where player~2, the maximizer, moves first and chooses a stopping rule without knowing $\theta$, and player~1, the minimizer, gets to choose a difficult $\theta$ with knowledge of the stopping rule. This max-min problem is, in fact, more difficult than the prophet problem for i.i.d.~random variables from an unknown distribution considered by \citet{CorreaDFS19}. 

Indeed, the construction of \citeauthor{CorreaDFS19} relies on the infinite version of Ramsey's theorem \citep{Rams30a}, and in particular leads to distributions with infinite support for which the order of minimization and maximization cannot be reversed.

Minimax theorems do exist that can handle finite strategy spaces of player 2 and compact metric strategy spaces for player 1~\citep[\eg][Proposition 1.17]{MSZ}, and this is the case we would get if the difficult instances for unknown i.i.d.~distributions would have a support that is finite and bounded by a number that depends only on $n$. The argument sketched above could thus be rescued through a variant of the construction of \citet{CorreaDFS19} with this property. 
We provide such a construction by using a \emph{finite} version of Ramsey's theorem as given for example by \citet{CFS10}.

An interesting aspect of our argument for the game theory connoisseur is that the minimax theorem of course requires mixed strategies.  This is clearly not a problem for player 2, but for player~1 this involves mixing over stopping rules. The validity of the above argument thus requires that any mixture of stopping rules can be implemented as a stopping rule. We will see that this readily follows from Kuhn's celebrated theorem on behavior strategies in extensive form games \citep{Kuhn53}.

\subsection{Further Related Work}

\paragraph{Prophet Inequalities and Pricing.} A comprehensive overview of early work on the classic single-choice prophet inequality can be found in the survey of \citet{HillKertz92}. Starting from the work of \citet{HajiaghayiKS07}, prophet inequalities have seen a surge of interest in the theoretical-computer-science literature. 

Two important directions have considered extensions to richer domains where more than one element can be selected~\citep[\eg][]{KlWe19a,Alaei14,DuettingK15,FeldmanGL15,DuettingFKL17,Rubinstein16,RubinsteinS17,ChawlaMT19,GravinW19,EzraFGT20,DuettingKL20}), and random- or best-order models, both in the single-item setting and in combinatorial domains~\citep[\eg][]{EHLM15a,AEEH+17a,ACK18a,EhsaniHKS18,CSZ19a,AgrawalSZ19}. 

Prophet inequalities with inexact priors were studied by \citet{DuettingK19}, and prophet inequalities for unknown distributions with or without access to samples by \citet{AzarKW14,BabaioffBDS17,CorreaDFS19,RubinsteinWW20,CorreaCES20,KaplanNR20}.

Exchangeable random variables are generally positively correlated. Apart from the classical work of \citet{rinott1987} on negatively correlated random variables, the only work we are aware of to explicitly study correlation in the context of prophet inequalities is that of \citet{ImmorlicaSW20}. 

Problems in pricing where distributional information is unavailable and must be learned have also been considered in operations research and management science. These problems typically differ significantly from those studied in computer science, and solutions often involve some form of regret minimization. Examples include recent results of \citet{GoZe17a} and those described in a survey of \citet{Boer15}.

\paragraph{Streaming Algorithms.}

To the best of our knowledge, we are the first to use the streaming model of computation to study prophet inequality problems. Related work has previously considered streaming algorithms for submodular maximization \citep[\eg][]{BadanidiyuruMKK14,FeldmanNSZ20}. The algorithms in this literature are typically not required to make irrevocable online decisions, but they also don't assume a generative model for the data.

More closely related is work in the streaming literature that concerns the computation of various aggregate statistics of large amounts of stochastic data. For example, most prophet inequalities rely on the mean or median, or on more fine-grained information about the underlying distributions such as quantiles. While the empirical mean and variance of a sequence of values can be updated efficiently~\citet{West79}, strong lower bounds exist for the estimation of the median and other quantiles~\citet{GuhaM09}.

Even more specifically, our streaming implementation of MRS algorithms shares certain characteristics with reservoir sampling and its variants~\citep[\eg][]{Vitter85,Li94}, where the goal is to produce at any point in time a random subset of size~$k$ of the values seen so far. Compared to reservoir sampling, however, we require subsets of varying size at certain points in time as well as a better space complexity.

\paragraph{Algorithms from Data.}

In analyzing MRS algorithm with and without streaming, we optimize over a class of algorithms that has limited information about the problem at hand. Related problems have been considered under the umbrella of application-specific algorithm selection and data-driven algorithm design \citep[\eg][]{GuptaR17,AilonCCLMS11,BalcanDV18,BalcanDW18}, in particular in the context of designing revenue-optimal auctions from samples~\citep[\eg][]{ColeR14,MorgensternR15}.

\section{Preliminaries}
\label{sec:prelims}

Denote by $\N$ the set of positive integers and let $\N_0=\N\cup\{0\}$. For $i\in\N$, let $[i]=\{1,\dots,i\}$.

\paragraph{Sequences of Random Variables.}

We consider sequences of non-negative and integrable random variables $X_1,\dots,X_n$ drawn independently from a distribution $F$. The distribution $F$ will either be unknown, or it will be drawn from a known prior distribution $\theta\in\Theta$, where $\Theta$ is the set of finite distributions over real distributions with finite support. A sequence of random variables is called exchangeable if their joint distribution is invariant under permutations of the sequence. The set of exchangeable sequences contains the sequences described above but is more general.

\paragraph{Stopping Rules.}

A stopping rule observes a sequence of random variables and decides to stop based on the random variables seen so far and possibly some additional information regarding the distribution from which the random variables are drawn. We consider two types of stopping rules, corresponding respectively to the cases where random variables are drawn from an unknown distribution and from a distribution that has itself been drawn according to a prior.

In the case of an unknown distribution, we consider stopping rules that in addition to the random variables observed so far may depend on $k$ additional independent samples $S_1,\dots,S_k$ from the same distribution as the random variables. 
Such a stopping rule can be expressed as a family $\vr$ of functions $r_1,\dots,r_n$, where $r_i:\R_+^{k+i}\rightarrow[0,1]$ for all $i=1,\dots,n$. Here, for any $\vs\in\R_+^k$ and $\vx\in\R_+^n$, $r_i(s_1\dots,s_k,x_1,\dots,x_i)$ is the probability of stopping at $X_i$ when we have observed samples $S_1=s_1\dots,S=s_k$ and values $X_1=x_1,\dots,X_i=x_i$ and have not stopped at $X_1,\dots,X_{i-1}$.
The \emph{stopping time} $\tau$ of such a stopping rule $\vr$ is the random variable with support $\{1,\dots,n\}\cup\{\infty\}$ such that for all $\vs\in\R_+^k$ and $\vx\in\R_+^n$,
\begin{align*}
	\Pr{\tau=i\mid S_1=s_1,\dots,S_k=s_k,X_1=x_1,\dots,X_n=x_n} = \Biggl(\prod_{j=1}^{i-1}\bigl(&1-r_j(s_1,\dots,s_k,x_1,\dots,x_j)\bigr)\Biggr)\\
	&\cdot r_i(s_1\dots,s_k,x_1,\dots,x_i) .
\end{align*}

In the case of a distribution drawn from a prior $\theta$, we consider stopping rules that in addition of the random variables observed so far may depend on $\theta$. 
Such a stopping rule can be expressed as a family $\vr$ of functions $r_1,\dots,r_n$, where $r_i:\R_+^{i}\times\Theta\rightarrow[0,1]$ for all $i=1,\dots,n$. Here, for any $\vx\in\R_+^n$ and $\theta\in\Theta$, $r_i(x_1,\dots,x_i,\theta)$ is the probability of stopping at $X_i$ when we have observed the values $X_1=x_1,\dots,X_i=x_i$, have not stopped at $X_1,\dots,X_{i-1}$, and when the prior distribution is $\theta$. 
The \emph{stopping time} $\tau$ of such a stopping rule $\vr$ is thus the random variable with support $\{1,\dots,n\}\cup\{\infty\}$ such that for all $\vx\in\R_+^n$ and $\theta\in\Theta$,
\begin{align*}
	\Pr{\tau=i\mid X_1=x_1,\dots,X_n=x_n, \theta}= & \Biggl(\prod_{j=1}^{i-1}\bigl(1-r_j(x_1,\dots,x_j,\theta)\bigr)\Biggr)
	\cdot r_i(x_1,\dots,x_i,\theta) .
\end{align*}

\paragraph{Prophet Inequalities.} 

For a given stopping rule we will be interested in the expected value $\E{X_\tau}$ of the variable at which it stops, where we use the convention that $X_{\infty}=0$, and will measure its performance relative to the expected maximum $\E{\max\{X_1,\dots,X_n\}}$ of the random variables $X_1,\dots,X_n$. We will say that a stopping rule achieves approximation guarantee~$\alpha$, for $\alpha\leq 1$, if for any distribution or prior over distributions, $\E{X_\tau}\geq\alpha\,\E{\max\{X_1,\dots,X_n\}}$.

\paragraph{Streaming Algorithms.}

In the case of an unknown distribution we will be interested specifically in stopping rules that can be implemented as streaming algorithms. We will assume that the stream consists of the samples $S_1,\dots,S_k$ followed by the values $X_1,\dots,X_n$. A streaming algorithm is then allowed a single pass over the sequence, and its space complexity is required to be logarithmic both in the length of the sequence, which will in fact be $O(n)$, and in $\max\{S_1,\dots,S_k,X_1,\dots,X_n\}$. Since the streaming algorithms we consider exclusively store values that occur in the stream, we will express their space complexity in terms of $n$ only.

\medskip

For ease of exposition we will assume continuity of distributions in proving lower bounds and use discrete distributions to prove upper bounds. All results can be shown to hold in general by standard arguments, to break ties among random variables and to approximate a discrete distribution by a continuous one.

\section{I.I.D.\@ Random Variables from an Unknown Distribution}
\label{sec:subset}

We now turn to the setting of unknown i.i.d.\@ random variables, which was first studied by \citet{CorreaDFS19}. Here a lower bound of $\alpha=1/e$ can be obtained via the optimal solution to the secretary problem, and this bound cannot be improved upon without access to additional samples from the distribution. A sharp phase transition occurs when there are linearly many samples in the number of variables. In particular, \citeauthor{CorreaDFS19} gave an algorithm that with~$n$ samples achieves an approximation ratio of $\alpha \geq 1-1/e \approx 0.6321$. The same bound was obtained by \citet{KaplanNR20}, and improved to $0.635$ by \citet{CorreaCES20}. We improve on this bound with~$n$ samples, and also obtain improved bounds for the case of $\beta n$ samples for variable $\beta>0$.

We do so by considering a natural class of algorithms that can be thought of as fixing a function $f:[n]\rightarrow\N$, and accepting variable $X_i$ if its value exceeds the maximum of $f(i)$ fresh samples from the distribution. In fact, we will not draw fresh samples at each time step but rather choose uniformly at random a subset of size $f(i)\leq k+i-1$ from the set containing the $k$ samples available to the algorithm at the outset and the values of the random variables $X_1,\dots,X_{i-1}$ observed so far. We will refer to algorithms that follow this general strategy as \emph{maximum-of-random-subset} (MRS) algorithms, and will see in \autoref{subsec:structural-lemma} that they behave as if they would in fact draw fresh samples at each time step. The bound of $1-1/e$ of \citeauthor{CorreaDFS19} can in fact be achieved with an MRS algorithm where for all $i$, $f(i)=n-1$.

We will first consider MRS algorithms that may use an arbitrary number~$k$ of samples, and give a tight bound on their approximation guarantee. Interestingly, the optimal such guarantee of $\alpha\approx 0.6534$ is obtained with a bounded number of samples, namely $k\approx 1.4434 \cdot n$. We then give bounds for MRS algorithms that use at most $\beta n$ samples for $0\leq\beta\leq 1.4434$, and show specifically that there exists an MRS algorithm using $n$ samples with an approximation guarantee of $\alpha\geq 0.6489>0.635 > 1-1/e$.

\begin{theorem}
	\label{thm:unconstrained}
	Consider a sequence of $n$ random variables $X_1,\dots,X_n$ drawn independently from an unknown distribution. Then, as $n\rightarrow\infty$, the best MRS algorithm with an unconstrained number $k$ of samples achieves an approximation guarantee of $\alpha\approx 0.6534$ and requires $k\approx 1.4434 \cdot n$ samples.
\end{theorem}

\begin{theorem}
	\label{thm:nsamples}
	Consider a sequence of $n$ random variables $X_1, \dots, X_n$ drawn independently from an unknown distribution. Then there exists an MRS algorithm that uses $k=n$ samples and achieves an approximation ratio of $\alpha\geq 0.6489$ as $n\rightarrow\infty$.
\end{theorem}

Towards proving these theorems, we first show a structural lemma in \autoref{subsec:structural-lemma} that allows us to analyze MRS algorithms as if they would draw fresh samples at each step. In \autoref{subsec:three-step}, we then consider as a warm-up an MRS algorithm using a simple piecewise linear function that improves upon the the approximation guarantee $1-1/e$. The remainder of the section is concerned with computing the best MRS algorithm. In \autoref{subsec:unconstrained}, we consider the setting in which an unbounded number of samples is available, leading to the proof of \autoref{thm:unconstrained}. In \autoref{subsec:constrained}, we then turn to the setting in which only a limited number of samples is available, proving \autoref{thm:nsamples} and additional bounds for other values of $k$.

\subsection{Definition and Structural Lemma}
\label{subsec:structural-lemma}

Let $k\in\N$ and consider a function $f:[n]\rightarrow\mathbb{N}$ where $f(i)\leq k+i-1$ for all $i\in[n]$. Then the MRS algorithm based on $f$ proceeds as follows: given that it arrives at random variable $X_i$, it selects a uniformly random subset $\mathcal{R}_i=\{R^1_i,\dots,R^{f(i)}_i\}$ of size $f(i)$ from the set $\{S_1,\dots,S_k,X_1,\dots,X_{i-1}\}$ of~$k$ samples and the first $i-1$ random variables and sets $\max\mathcal{R}_i$ as threshold for $X_i$. We have the following lemma.

\begin{lemma}\label{lem:bayes}
	Consider some MRS algorithm based on $f:[n]\rightarrow\mathbb{N}$ and $i\in[n]$. Conditioned on the fact that the algorithm arrives at step $i$, the distribution of the set $\{S_1,\dots,S_k,X_1,\dots,X_{i-1}\}$ of values seen before step $i$ is identical to the distribution of a set of $k+i-1$ fresh samples from $F$.
\end{lemma}
\begin{proof}
	We show the claim by induction on~$i$, and start by observing that it clearly holds for $i=1$. Now suppose the claim holds for $i=1,\dots,i^\star-1$. Then, conditioned on the fact that the algorithm arrives at step $i^\star-1$, the set $\mathcal{T}=\{S_1,\dots,S_k,X_1,\dots,X_{i^\star-2}\}$ has the same distribution as the one of a set of $k+i^\star-2$ fresh samples, so the distribution of the set $\mathcal{T}'=\{S_1,\dots,S_k,X_1,\dots,X_{i^\star-1}\}$ is the same as the one of a set of $k+i^\star-1$ fresh samples. We will argue that the decision of the algorithm to stop at $X_{i^\star-1}$ or to continue does not depend on the realization of $\mathcal{T}'$, which implies the claim.
	
	Since $F$ is continuous, we may assume that all the values $S_1,\dots,S_k,X_1,\dots,X_{i^\star-1}$ are distinct, so that each of these values can be identified with a unique rank in~$[k+i^\star-1]$. By definition, the decision of an MRS algorithm to stop or continue only depends on the ranks of the values $R^1_{i^\star-1},\dots,R^{f(i^\star-1)}_{i^\star-1},X_{i^\star-1}$.
	 Since the distribution of $\mathcal{T}$, from which $R^1_{i^\star-1},\dots,R^{f(i^\star-1)}_{i^\star-1}$ are drawn, is that of $k+i^\star-1$ fresh samples, and since $X_{i^\star-1}$ \emph{is} a fresh sample, those ranks are $f(i^\star-1)+1$ uniform draws without replacement from~$[k+i^\star-1]$ and thus independent of the realization of~$\mathcal{T}'$.
\end{proof}

Note that this lemma immediately implies that a set of size $f(i)$ selected uniformly at random from $\{S_1,\dots,S_k,X_1,\dots,X_{i-1}\}$ is also distributed like~$f(i)$ fresh samples from~$F$. This insight is what will enable us to give a mathematical expression for the value obtained by an MRS algorithm.

\subsection{Warm-Up: Three-Step Functions}
\label{subsec:three-step}

As a warm-up, let us convince ourselves that an MRS algorithm with access to $k=n$ samples can improve on the bound of $1-1/e$. We will assume for simplicity that~$n$ is a multiple of~$3$, but note that the conclusion holds in general as $n\to\infty$. Let $f:[n]\to\N$ with $f(i)=n-1$ for $i=1,\dots,\frac{n}{3}$, $f(i)=4n/3-1$ for $i=\frac{n}{3}+1,\dots,\frac{2n}{3}$, and $f(i)=2n/3-1$ for $i=\frac{2n}{3}+1,\dots,n$. Observe that the MRS algorithm for~$f$ uses only~$n$ samples. Consider $X_1,\dots,X_n$ drawn independently from an arbitrary distribution~$F$, and let~$\tau$ be the stopping time of the algorithm for these random variables. 
Then, 
\begin{align*}
	\mathbb{E}\left[\max\{X_1,\dots,X_n\}\right] = \int_{0}^{\infty} 1- F(x)^n \; \mathrm{d}x
\end{align*}
and, by \autoref{lem:bayes}, 
\begin{align*}
	\mathbb{E}[X_\tau] = & \underbrace{\sum_{i=1}^{n/3} \left(1-\frac{1}{n}\right)^{i-1}\frac{1}{n}}_{=: T_1} \int_{0}^{\infty} 1-F(x)^n \;\mathrm{d}x + {} \\
	& (1-T_1) \underbrace{\sum_{i=1}^{n/3} \left(1-\frac{3}{4n}\right)^{i-1}\frac{3}{4n}}_{=:T_2} \int_{0}^{\infty} 1-F(x)^{4n/3} \; \mathrm{d}x + {} \\
	&(1-T_1)(1-T_2) \underbrace{\sum_{i=1}^{n/3} \left(1-\frac{2}{3n}\right)^{i-1}\frac{2}{3n}}_{=:T_3} \int_{0}^{\infty} 1-F(x)^{2n/3} \; \mathrm{d}x .
\end{align*}
For $n\rightarrow\infty$, $T_1\rightarrow 1-1/e^{1/3}$, $T_2\rightarrow 1-1/e^{1/4}$, and $T_3\rightarrow 1-1/e^{1/2}$, and thus
\begin{align*}
\mathbb{E}[X_\tau] = \left(1-\frac{1}{e^{1/3}}\right) \int_{0}^{\infty} 1-F(x)^n \; \mathrm{d}x + {} & \frac{1}{e^{1/3}} \left(1-\frac{1}{e^{1/4}}\right) \int_{0}^{\infty} 1-F(x)^{4n/3} \; \mathrm{d}x + {} \\
& \frac{1}{e^{1/3}} \frac{1}{e^{1/4}} \left(1-\frac{1}{e^{1/2}}\right) \int_{0}^{\infty} 1-F(x)^{2n/3} \; \mathrm{d}x.
\end{align*}
To show that $\mathbb{E}[X_\tau]\geq\alpha\mathbb{E}[\max\{X_1,\dots,X_n\}]$, it suffices to show that for all $a\in[0,1]$,
\begin{align*}
\left(1-\frac{1}{e^{1/3}}\right) (1-a) + \frac{1}{e^{1/3}} \left(1-\frac{1}{e^{1/4}}\right) \left(1-a^{4/3}\right) + \frac{1}{e^{1/3}} \frac{1}{e^{1/4}} \left(1-\frac{1}{e^{1/2}}\right) \left(1-a^{2/3}\right) \geq \alpha (1-a).
\end{align*}
By dividing both sides of the inequality by $1-a$ and minimizing the resulting left-hand side over $a\in[0,1)$, we see that this is the case for $\alpha\approx 0.6370>1-1/e$.

As we will see later in \autoref{subsec:control-problem}, requiring the approximation ratio to hold pointwise for each value of the variable of integration rather than just for the sum of integrals is without loss. We will thus be able to obtain tight bounds using this technique.

\subsection{Proof of \autoref{thm:unconstrained}}
\label{subsec:unconstrained}

To prove \autoref{thm:unconstrained} we would like to find the best possible choice of $f:[n]\to\N$, without any restriction on the number of samples it is allowed to use. For a particular choice of~$f$, we will again determine the approximation ratio~$\alpha$ of the corresponding MRS algorithm by minimizing the ratio between $\mathbb{E}[X_\tau]$ and $\mathbb{E}[\max\{X_1,\dots,X_n\}]$. We will again express both in terms of integrals, and turn the minimization over distributions~$F$ into a minimization over~$a\in[0,1)$ by requiring the approximation ratio to hold pointwise for each value of the variable of integration. Maximization over~$f$ and minimization over~$a$ will yield a max-min control problem, which we will solve optimally.

\subsubsection{Formulation as a Control Problem}  \label{subsec:control-problem}
 
Fix $n\in\N$, and consider an MRS algorithm given by the function $f:[n]\rightarrow\mathbb{N}$. We can construct a continuous function $g:[0,1]\rightarrow\mathbb{R}_+$ from $f$ by setting $g(i/n)=f(i)/n$ for all $i\in[n]$ and linearly interpolating between these values. Similarly, if we were only given a continuous function $g$ in the first place, we could obtain $f$ from $g$ by setting $f(i)=\lceil g(i/n)\cdot n\rceil$ for all $i\in[n]$. In what follows we will compute the optimal such function $g$ and thereby the optimal MRS algorithm for all values of $n$.
To do so, consider a sequence $X_1,\dots,X_n$ of random variables drawn i.i.d.\@ from a distribution~$F$, and denote the stopping time of the MRS algorithm on this sequence by~$\tau$. Then
\begin{align*}
	\E{X_\tau} &= \sum_{i=1}^n\Pr{\text{$A$ arrives at $X_i$}}\cdot\Pr{\text{$A$ accepts $X_i$}\mid\text{$A$ arrives at $X_i$}}\cdot\E{X_i\mid\text{$A$ accepts $X_i$}} \\
	&= \sum_{i=1}^n\prod_{j=1}^{i-1}\left(1-\frac{1}{\lceil g(\frac{j}{n})\cdot n\rceil+1}\right)\cdot \frac{1}{\lceil g(\frac{i}{n})\cdot n\rceil+1}\cdot \int_0^\infty \! \bigl(1-F^{\lceil g(\frac{i}{n})\cdot n\rceil+1}(x)\bigr)\;\mathrm{d}x\\
	&= \sum_{i=1}^n\exp\left(\sum_{j=1}^{i-1}\ln\left(1-\frac{1}{\lceil g(\frac{j}{n})\cdot n\rceil+1}\right)\right)\cdot \frac{1}{\lceil g(\frac{i}{n})\cdot n\rceil+1}\cdot \int_0^\infty \! \bigl(1-F^{\lceil g(\frac{i}{n})\cdot n\rceil+1}(x)\bigr)\;\mathrm{d}x \\
	&=\sum_{i=1}^n\exp\biggl(-\sum_{j=1}^{i-1}\biggl(\frac{1}{\lceil g(\frac{j}{n})\cdot n\rceil+1}-O\Bigl(\frac{1}{n^2}\Bigr)\biggr)\biggr)\cdot \frac{1}{\lceil g(\frac{i}{n})\cdot n\rceil+1}\cdot \int_0^\infty \! \bigl(1-F^{\lceil g(\frac{i}{n})\cdot n\rceil+1}(x)\bigr)\;\mathrm{d}x\\
	&= e^{-O(\frac1n)} \sum_{i=1}^n\exp\Biggl(-\sum_{j=1}^{i-1}\biggl(\frac{1}{\lceil g(\frac{j}{n})\cdot n\rceil+1}\biggr)\Biggr)\cdot \frac{1}{\lceil g(\frac{i}{n})\cdot n\rceil+1}\cdot \int_0^\infty \! \bigl(1-F^{\lceil g(\frac{i}{n})\cdot n\rceil+1}(x)\bigr)\;\mathrm{d}x ,
\end{align*}
where for the fourth equality we have used that the Laurent series of $\ln(1-\frac{1}{x})-(-\frac{1}{x})$ at $x=\infty$ is $\sum_{i=2}^\infty -\frac{x^{-i}}{i}=-O(\frac{1}{x^2})$ (and assumed w.l.o.g.\ that $g$ is bounded away from $0$).   
Thus, for $n\rightarrow\infty$,
\begin{align*}
	\E{X_\tau}=&\int_0^1\exp\left(-\int_0^y\frac{1}{g(z)}\;\mathrm{d}z\right)\cdot\frac{1}{g(y)}\cdot\int_0^\infty\left(1-F^{g(y)\cdot n}(x)\right)\mathrm{d}x\;\mathrm{d}y\\
	=&\int_0^\infty\int_0^1\exp\left(-\int_0^y\frac{1}{g(z)}\;\mathrm{d}z\right)\cdot\frac{1}{g(y)}\cdot\left(1-F^{g(y)\cdot n}(x)\right)\mathrm{d}y\;\mathrm{d}x,
\end{align*}
where we exchange the order of integration in the second step using Fubini's theorem, which may be applied, because the integrand is clearly positive.

Our goal is to find the maximum value $\alpha\in\mathbb{R}_+$ for which $\E{X_\tau}\geq\alpha\cdot\E{\max\{X_1,\dots,X_n\}}$ or, equivalently,
\begin{equation}\label{eq:MRS-goal}
	\int_0^\infty\int_0^1\exp\left(-\int_0^y\frac{1}{g(z)}\;\mathrm{d}z\right)\cdot\frac{1}{g(y)}\cdot\left(1-F^{g(y)\cdot n}(x)\right)\mathrm{d}y\;\mathrm{d}x\geq \int_0^\infty \alpha\cdot\left(1-F^n(x)\right)\;\mathrm{d}x .
\end{equation}
A sufficient condition for the latter is that for all $a\in[0,1]$,
\begin{equation}\label{eq:MRS-goal-pointwise}
	\int_0^1\exp\left(-\int_0^y\frac{1}{g(z)}\;\mathrm{d}z\right)\cdot\frac{1}{g(y)}\cdot\left(1-a^{g(y)}\right)\mathrm{d}y\geq\alpha\cdot(1-a),
\end{equation}
and this condition is in fact also necessary. Indeed, if~\eqref{eq:MRS-goal-pointwise} is violated for some $\alpha$ and $a$, then~\eqref{eq:MRS-goal} is violated for $\alpha$ and the cumulative distribution function~$F$ of a random variable that has value~$0$ with probability~$a$ and value~$1$ with probability $(1-a)$. This choice of~$F$ makes the integrand on the right-hand side of~\eqref{eq:MRS-goal} greater than the integrand on the left-hand side for all~$x$ for which the integrands are non-zero, \ie for all $x<1$, thus violating~\eqref{eq:MRS-goal}.

To determine the approximation ratio of the MRS algorithm~$A$ we can thus focus on finding the maximum value~$\alpha$ such that~\eqref{eq:MRS-goal-pointwise} is satisfied for all $a$. Since~\eqref{eq:MRS-goal-pointwise} is trivially satisfied for $a=1$, we are interested in the optimum value of the control problem 
\begin{align} \label{eq:controlp}
	\mathcal{P} &= \sup_{g:[0,1]\rightarrow \R_+}\inf_{a\in[0,1)}\left\{\int_0^1\exp\left(-\int_0^y\frac{1}{g(z)}\;\mathrm{d}z\right)\cdot\frac{1-a^{g(y)}}{g(y)\cdot(1-a)}\; \mathrm{d}y\right\} \notag\\
	&= \sup_{\substack{h:[0,1]\rightarrow\mathbb{R}_+,\\h(0)=0}}\inf_{a\in[0,1)}\left\{\int_0^1e^{-h(y)}\cdot h'(y)\cdot\frac{1-a^{\frac{1}{h'(y)}}}{1-a}\;\mathrm{d}y\right\},
\end{align}
where the second equality can be seen to hold by choosing $h:[0,1]\rightarrow\mathbb{R}_+$ such that $h(y)=\int_0^y\frac{1}{g(z)}\;\mathrm{d}z$ for all $y\in[0,1]$, which implies that $g(y)=\frac{1}{h'(y)}$.

\subsubsection{Solving the Control Problem}
\label{subsec:upperbound}
\label{subsec:lowerbound}

We solve the control problem~$\mathcal{P}$ by giving matching upper and lower bounds. For the upper bound we swap supremum and infimum and apply the Euler--Lagrange equation to the supremum, which is now the inner problem, to write any optimal function~$h$ in terms of~$a$ and a single parameter~$\mu$. We then guess the value of~$a$ at which the infimum is attained and solve the remaining supremum over~$\mu$. For the lower bound we replace~$h$ by its parametric form, guess the values of the parameters at which the supremum is attained, and solve the remaining infimum over~$a$. In both cases we obtain the same value of approximately $0.6534$. Inspection of the optimal function~$h$ reveals that it is non-increasing, which implies that $g(0)\cdot n=\frac{1}{h'(0)}n\approx 1.4434\cdot n$ samples are sufficient to implement the optimal MRS algorithm.

It is worth pointing out that the change of the order of supremum and infimum and the substitution of a particular form of~$h$ are potentially lossy but turn out to be without loss. This means that a minimax theorem holds for $\mathcal{P}$, and that a universal worst-case distribution~$F$ exists that applies to all MRS algorithms.

\paragraph{Upper Bound.}

By the max-min inequality,
\begin{align}  \label{eq:infsup}
	\mathcal{P} \leq \inf_{a\in[0,1)}\sup_{\substack{h:[0,1]\rightarrow\mathbb{R}_+,\\h(0)=0}}\left\{\int_0^1e^{-h(y)}\cdot h'(y)\cdot\frac{1-a^{\frac{1}{h'(y)}}}{1-a}\;\mathrm{d}y\right\}.
\end{align}
Now the inner problem can be written as
\begin{align*}
\sup_{\substack{h:[0,1]\rightarrow\mathbb{R}_+,\\h(0)=0}} \int_0^1 L(y,h(y),h'(y))\;\mathrm{d}y ,
\end{align*}
where
\[
	L(y,h(y),h'(y))=e^{-h(y)}\cdot h'(y)\cdot\frac{1-a^{\frac{1}{h'(y)}}}{1-a} .
\]
A necessary condition for optimality of $h$ is the Euler--Lagrange equation
\begin{align}
	\frac{\partial}{\partial h} L(y,h(y),h'(y)) - \frac{\mathrm{d}}{\mathrm{d}y}\;\frac{\partial}{\partial h'} L(y,h(y),h'(y)) = 0 , \label{eq:euler-lagrange}
\end{align}
where
\begin{align}
	\frac{\partial}{\partial h}\;L(y,h(y),h'(y))=-e^{-h(y)}\cdot h'(y)\cdot\frac{1-a^{\frac{1}{h'(y)}}}{1-a} \label{eq:part1}
\end{align}
and
\begin{align}
	\frac{\mathrm{d}}{\mathrm{d}y}\;\frac{\partial}{\partial h'}\;L(y,h(y),h'(y))=&\frac{\mathrm{d}}{\mathrm{d}y}\left(e^{-h(y)}\cdot\frac{1-a^{\frac{1}{h'(y)}}}{1-a}+e^{-h(y)}\cdot h'(y)\cdot\frac{\ln a\cdot a^{\frac{1}{h'(y)}}}{(1-a)\cdot (h'(y))^2}\right) \notag\\
	=&{}-e^{-h(y)}\cdot h'(y)\cdot\frac{1-a^{\frac{1}{h'(y)}}}{1-a}+e^{-h(y)}\cdot\frac{\ln a\cdot a^{\frac{1}{h'(y)}}\cdot h''(y)}{(1-a)\cdot(h'(y))^2} \notag\\
	&{}+e^{-h(y)}\cdot(h''(y)-(h'(y))^2)\cdot\frac{\ln a\cdot a^{\frac{1}{h'(y)}}}{(1-a)\cdot(h'(y))^2} \notag\\ 
	&{}+e^{-h(y)}\cdot h'(y)\cdot\left(-\frac{(\ln(a))^2 a^\frac{1}{h'(y)} h''(y)}{(1-a)(h'(y))^4}-\frac{2 \ln(a) a^\frac{1}{h'(y)} h''(y)}{(1-a)(h'(y))^3}\right).	
	\label{eq:part2}
\end{align}
Substitution of~\eqref{eq:part1} and~\eqref{eq:part2} into~\eqref{eq:euler-lagrange} and simplification yields that
\begin{align*}
	-e^{-h(y)} \frac{\ln(a) a^\frac{1}{h'(y)}}{1-a} - e^{-h(y)} \frac{(\ln(a))^2 a^\frac{1}{h'(y)} h''(y)}{(1-a) (h'(y))^3} &= 0.
\end{align*}
Since $e^{x}>0$ for all $x$ and $1-a>0$ for $a\in[0,1)$, an equivalent requirement is that 
\[
	-\frac{h''(y)}{(h'(y))^3} = \frac{1}{\ln(a)}
\]
with the boundary condition $h(0)=0$.

Solving this second-order nonlinear ordinary differential equation yields two classes of parametric solutions
\begin{align*}
h_1(y) &= \sqrt{\kappa-\mu y} - \sqrt{\kappa}, \quad h_2'(y) = - \frac{\mu}{2\sqrt{\kappa-\mu y}}, \quad \mu = - 2 \ln(a) \geq 0, \quad \kappa \geq \mu, \quad\text{and}\\
h_2(y) &= \sqrt{\kappa}  - \sqrt{\kappa-\mu y}, \quad h_2'(y) = \textcolor{white}{-}\frac{\mu}{2\sqrt{\kappa-\mu y}}, \quad \mu = - 2 \ln(a) \geq 0, \quad \kappa \geq \mu,
\end{align*}
where only the latter guarantees that $g(y) = 1/h'(y) \geq 0$.

\newcommand{\muopt}{\bar\mu}
\newcommand{\aopt}{\bar{a}}
Let $\muopt\approx 1.9202$ be the unique value such that
\[
	1 - \frac{e^{\sqrt{\muopt}}}{\sqrt{\muopt}} +\frac{e^{\frac{\muopt}{2}}}{\sqrt{\muopt}} = 0,
\]
and $\aopt=e^{-\frac{\muopt}{2}}\approx 0.3829$.

By setting $h=h_2$ and $a=\aopt$ in \eqref{eq:infsup}, and showing that the remaining supremum over~$\kappa$ is attained for~$\kappa=\muopt$, we conclude that
\[
	\mathcal{P} \leq \frac{e^{-\sqrt{\muopt}}(1-e^{\sqrt{\muopt}}+\sqrt{\muopt})}{e^{-\frac{\muopt}{2}}-1} \approx 0.6534 .
\]

\paragraph{Lower Bound.}

By restricting the supremum in~\eqref{eq:controlp} to functions of the form $h(y)=\sqrt{\mu}-\sqrt{\mu\cdot(1-y)}$ for some $\mu\in\mathbb{R}_+$, which satisfy the boundary condition that $h(0)=0$, we see that
\begin{align*}
	\mathcal{P} \geq &\sup_{\mu\in\mathbb{R}_+}\inf_{a\in[0,1)}\left\{\frac{e^{-\sqrt{\mu}}}{1-a}\cdot\int_0^1\frac{e^{\sqrt{\mu \cdot(1-y)}}\cdot\mu\cdot\left(1-a^{\frac2\mu\sqrt{\mu \cdot(1-y)}}\right)}{2\cdot\sqrt{\mu \cdot(1-y)}}\;\mathrm{d}y\right\}\\
	=&\sup_{\mu\in\mathbb{R}_+}\inf_{a\in[0,1)}\left\{\frac{e^{-\sqrt{\mu}}}{1-a}\cdot\left[\frac{e^{\sqrt{\mu\cdot(1-y)}\cdot(1+\frac{2\ln a}{\mu})}}{1+\frac{2\ln a}{\mu}}-e^{\sqrt{\mu\cdot(1-y)}}\right]_0^1\right\}\\
	=&\sup_{\mu\in\mathbb{R}_+}\inf_{a\in[0,1)}\left\{\frac{e^{-\sqrt{\mu}}}{1-a}\cdot\left(\frac{1}{1+\frac{2\ln a}{\mu}}-1-\frac{e^{\sqrt{\mu}\cdot(1+\frac{2\ln a}{\mu})}}{1+\frac{2\ln a}{\mu}}+e^{\sqrt{\mu}}\right)\right\}\\
	=&\sup_{\mu\in\mathbb{R}_+}\inf_{\substack{b\in\mathbb{R}_+;\\b\notin\{0,1\}}}\left\{\frac{e^{-\sqrt{\mu}}}{1-e^{-\frac{\mu b}{2}}}\cdot\left(\frac{1}{1-b}-1-\frac{e^{\sqrt{\mu}\cdot(1-b)}}{1-b}+e^{\sqrt{\mu}}\right)\right\},	
\end{align*}
where the last equality can be seen to hold by setting $b=-\frac{2\ln a}{\mu}$ and $a=e^{-\frac{\mu b}{2}}$. 

By setting $\mu=\muopt$ in the last expression and showing that the remaining infimum over~$b$ is attained for $b\to 1$, we conclude that
\[
	\mathcal{P} \geq \frac{e^{-\sqrt{\muopt}}(1-e^{\sqrt{\muopt}}+\sqrt{\muopt})}{e^{-\frac{\muopt}{2}}-1} ,
\]
which equals the upper bound.

\medskip

The resulting optimal choice of~$g$, given by~$g(y)=1/h'(y)=2\sqrt{\muopt-\muopt y}/\muopt$, is non-increasing in~$y$ and thus has a maximum value of $g(0)=2/\sqrt{\muopt}\approx 1.4434$. This means that the optimal MRS algorithm can be implemented with slightly fewer than $3n/2$ samples.

\subsection{Proof of \autoref{thm:nsamples}}
\label{subsec:constrained}

We finally consider MRS algorithms that have access to $\beta n$ samples for some $\beta<2/\sqrt{\muopt}$. This imposes the constraint that $g(y)\leq \beta+y$ for all $y\in[0,1]$, and since the optimal MRS algorithm for the unconstrained case uses more than $\beta n$ samples the constraint must bind for some non-empty subset of $[0,1]$. To obtain a lower bound on the performance of the best MRS algorithm we may in fact assume that the constraint binds on $[0,t]$ for some $t\in[0,1]$, such that $g(y)=\beta+y$ and $h(y)=\int_{0}^{y} 1/g(z) \; dz = \ln(\beta+y)-\ln(\beta)$ for all $y\in[0,t]$. Proceeding as in \autoref{subsec:upperbound}, we can write the performance of the best MRS algorithm from the restricted class as a control problem
\begin{align*}
\mathcal{Q} &= \sup_{\substack{t\in[0,1],\\h:[t,1]\rightarrow\mathbb{R}_+,\\h(t) = \ln(\beta+t)}} \inf_{a \in [0,1)} \left\{\frac{\int_{0}^{t} \frac{\beta}{(\beta+y)^2}\cdot\left(1-a^{\beta+y}\right) \;\mathrm{d}y + \int_{t}^{1} e^{-h(y)} \cdot h'(y) \cdot \left(1-a^\frac{1}{h'(y)}\right)\; \mathrm{d}y}{1-a}\right\} \\
	&= \sup_{\substack{t\in[0,1],\\h:[t,1]\rightarrow\mathbb{R}_+,\\h(t) = \ln(\beta+t)}} \inf_{a \in [0,1)} \left\{\int_{0}^{t} \frac{\beta(1-a^{\beta+y})}{(\beta+y)^2(1-a)} \;\mathrm{d}y + \int_{t}^{1} e^{-h(y)} \cdot h'(y) \cdot \frac{1-a^\frac{1}{h'(y)}}{1-a}\;\mathrm{d}y\right\}.
\end{align*}

Note that the objective is now a sum of two integrals. The first integral is constant with respect to~$h$. The second integral has the same integrand as the integral in problem $\mathcal{P}$ from \autoref{subsec:upperbound}, but it begins at~$t$ rather than~$0$ and involves a function~$h$ that is subject to a different boundary condition, $h(t)=\ln(\beta+t)-\ln(\beta)$ instead of $h(0)=0$. As our application of the Euler--Lagrange equation in \autoref{subsec:upperbound} relied neither on the limits of integration nor on the boundary condition we obtain the same differential equation as before, $-h''(y)/(h'(y))^3 = 1/\ln(a)$, but subject to the new boundary condition that $h(t)=\ln(\beta+t)-\ln(\beta)$. 

Since $g(y)=1/h'(y)$ for $y\in(0,1)$ and thus
\[
	g(y)\cdot g'(y) = \frac{((g(y))^2)'}{2} = \frac{1}{2} \left(\frac{1}{(h'(y))^2}\right)' = - \frac{h''(y)}{(h'(y))^3}
\]
for $y\geq t$, we can alternatively solve the first-order non-linear differential equation $g(t)\cdot g'(t)=1/\ln(a)$. From the requirement that $g(y)\geq 0$ for all~$y$ we conclude that 
\begin{align*}
	g(y) = \sqrt{2}\cdot\sqrt{\frac{1}{\ln(a)}\cdot y + \kappa}
\end{align*}
for some $\kappa\geq -1/\ln(a)$, and by choosing $\kappa$ to satisfy the boundary condition that $g(t)=\beta+t$ we obtain
\begin{align*}
	g(y) = \sqrt{2} \cdot \sqrt{\frac{1}{\ln(a)}\cdot y + \frac{1}{2}\left(-\frac{2}{\ln(a)}\cdot t + t^2 + 2\beta t + \beta^2\right)} .
\end{align*}

In analogy to \autoref{subsec:lowerbound} we may derive a lower bound on the value of~$\mathcal{Q}$ by considering the parametric class of functions
\begin{align*}
	g(y) &= \sqrt{2} \cdot \sqrt{cy+\frac{1}{2}\left(-2ct+t^2+2\beta t+\beta^2\right)},
\end{align*}
where $c\leq 0$, and we may in fact choose $c=(\beta+t)^2/(2(t-1))$ such that $g(1)=0$ as before. Then 
\[
	g(y) = (\beta+t)\sqrt{\frac{y-1}{t-1}}
\]
and
\[
	h(y) = \ln(\beta +t) - \ln(\beta ) + \frac{2(y-1)}{\sqrt{\frac{(y-1) (\beta +t)^2}{t-1}}} - \frac{2 (t-1)}{\sqrt{(\beta +t)^2}}	
\]

\begin{table}
\begin{center}
\begin{tabular}{*{4}{c}}
\toprule
$\beta$ & $\alpha\geq{}$ & $t\approx{}$ & $a\approx{}$ \\ 
\midrule
1.4 & 0.653368 & 0.025503 & 0.383230 \\
1.3 & 0.653280 & 0.087540 & 0.387562 \\
1.2 & 0.652853 & 0.155180 & 0.398509 \\
1.1 & 0.651654 & 0.230674 & 0.419390 \\
1.0 & 0.648957 & 0.317590 & 0.455588 \\
0.9 & 0.643563 & 0.421611 & 0.515673 \\
0.8 & 0.633580 & 0.551596 & 0.612066 \\
\vdots &\vdots &\vdots &\vdots\\
\bottomrule
\end{tabular}
\hspace{.02\textwidth}
\begin{tabular}{*{4}{c}}
\toprule
$\beta$ & $\alpha\geq{}$ & $t\approx{}$ & $a\approx{}$ \\ 
\midrule
\vdots &\vdots &\vdots &\vdots\\
0.7 & 0.616281 & 0.720814 & 0.758359 \\
0.6 & 0.588379 & 0.949784 & 0.959047 \\
0.5 & 0.549306 & 1.000000 & 1.000000 \\
0.4 & 0.501105 & 1.000000 & 1.000000 \\
0.3 & 0.439901 & 1.000000 & 1.000000 \\
0.2 & 0.358351 & 1.000000 & 1.000000 \\
0.1 & 0.239789 & 1.000000 & 1.000000 \\
\bottomrule
\end{tabular}
\caption{Lower bounds on the performance of the optimal MRS algorithm with access to $\beta n$ samples for varying values of $\beta$. The bounds arise as the minimum over~$a$ of a function in~$t$, and the values of~$t$ and~$a$ corresponding to each bound are given alongside it.}
\label{tab:samples-bound}
\end{center}
\end{table}

We can now substitute~$h$ into $\mathcal{Q}$ and solve the integrals to obtain a simpler control problem with a supremum over~$t$ and an infimum over~$a$. While we cannot solve this problem exactly, we may conjecture in analogy to problem~$\mathcal{P}$ that for the optimal choice of~$t$ the infimum over~$a$ is attained for $a\to e^{2(t-1)/(\beta+1)^2}$. We can then determine the value of~$t$ for which the conjectured infimum is smallest, which turns out to be unique, and obtain a lower bound on~$\mathcal{Q}$ and thus on the approximation guarantee of the best MRS algorithm by substituting this value of~$t$ into the simplified control problem and solving the remaining minimization problem over~$a$.

\autoref{tab:samples-bound} shows a selection of bounds obtained in this way for different values of~$\beta$, along with the choice of~$t$ that leads to each bound and the corresponding optimal choice of~$a$. The lower bounds are also shown graphically in \autoref{fig:samples-bound}. For $\beta=1$ in particular we obtain a lower bound of $\alpha\geq 0.6489$.

We conjecture that the bounds shown in \autoref{tab:samples-bound} are in fact tight up to errors in the numerical approximation. More specifically, we believe that it is without loss of generality to assume that $g(y)=\beta+y$ for all $y<t$ and some $t\in[0,1]$, $g(t)=\beta+t$, and $g(1)=0$.

\subsection{A Tight Bound for At Most $n/(e-1)$ Samples}

Comparison to the upper bound of \citet{CorreaDFS19} reveals that the bound in the previous section is in fact tight when $\beta=1/(e-1)$. The optimal MRS algorithm in this case stops at the first value that exceeds all samples, a behavior that should remind us of the optimal algorithm for the case $\beta=0$, which is identical to the optimal algorithm for the secretary problem. Indeed, $\beta=1/(e-1)$ implies that $\beta/(\beta+1)=1/e$, so the only difference to the case $\beta=0$ is that the algorithm does not need to skip any values because it is given the correct number of samples for free.

For the intermediate case where $0\leq\beta\leq 1/(e-1)$, we may now conjecture that we should skip values until the combined number of samples and skipped values amounts to a (rounded) $1/e$ fraction of the combined number of samples and values, and stop at the first value thereafter that exceeds all samples and values seen so far. It turns out that this algorithm matches the upper bound of \citet{CorreaDFS19}, and is thus optimal, when $0\leq\beta\leq 1/(e-1)$ and $n\rightarrow\infty$. Note that the limit is needed in our analysis to account for possible losses incurred when rounding.
\begin{theorem}
Let $\beta\leq 1/(e-1)$, $n\in\N$, $k=\lfloor\beta n\rfloor$, and $m=\lfloor(\frac{1+\beta}{e}-\beta)n\rfloor$. Consider a sequence of i.i.d.\@ random variables $S_1,\dots,S_k,X_1,\dots,X_n$. Let $\tau$ be the stopping time of the algorithm that stops at $X_j$ if (i)~it has not stopped previously, (ii)~$j>m$, and (iii)~$X_j>\max\{S_1,\dots,S_k,X_1,\dots,X_{j-1}\}$. Then, as $n\to\infty$, $\E{X_\tau}\geq \frac{1+\beta}{e}\cdot \E{\max\{X_1,\dots,X_n\}}$.
\end{theorem}
\begin{proof}
	To simplify the exposition, we will assume that $\frac{1+\beta}{e}-\beta$ and $\beta n$ are integers and drop the symbols $\lfloor$ and $\rfloor$. Let $\delta=\frac{1+\beta}{e}-\beta$ and $x\in\m{R}_+$. We will show that for large $n$ and uniformly in $F$,
\begin{equation} \label{dom}
\text{Pr}(X_{\tau} \geq x) \geq \left(\frac{1+\beta}{e}+o(1)\right) \cdot \text{Pr}(\max \left\{X_1,\dots,X_n\right\} \geq x). 
\end{equation}
The theorem then follows by integrating over $x$. 

We have
\begin{eqnarray*}
\text{Pr}(X_{\tau} \geq x)&=& \sum_{i=\delta n+1}^n \text{Pr}(\left\{X_{i} \geq x\right\} \cap \left\{\tau=i\right\}),
\end{eqnarray*}
and for $i \in \left\{\delta n+1, \dots, n \right\}$, $\text{Pr}(\left\{X_{i} \geq x\right\} \cap \left\{\tau=i\right\})$ is equal to
\begin{align*}
&\phantom{=\cdot}\text{Pr}(\left[X_{i} \geq \max \left\{x,S_1,\dots,S_{\beta n},X_1,\dots, X_{i-1}
\right\}\right] 
\\
&\phantom{=\cdot\text{Pr}(}\cap
  \left[\max\left\{S_1,\dots,S_{\beta n},X_1,\dots, X_{\delta n} \right\} \geq \max\left\{X_{\delta n+1},\dots,X_{i-1} \right\} \right])
\\
&=\phantom{\cdot}\text{Pr}\left(X_{i} \geq \max \left\{x,S_1,\dots,S_{\beta n},X_1,\dots, X_{i-1}
\right\}\right)
\\
&\phantom{=}\cdot
\text{Pr}(\max\left\{S_1,\dots,S_{\beta n},X_1,\dots, X_{\delta n} \right\} \geq \max\left\{X_{\delta n+1},\dots,X_{i-1} \right\}).
\end{align*}
We have
\begin{equation*}
\text{Pr}(\max\left\{S_1,\dots,S_{\beta n},X_1,\dots, X_{\delta n}\right\} )
\geq
 \max\left\{X_{\delta n+1},\dots,X_{i-1}\right\} )
=\frac{\beta +\delta }{\beta +\frac{i-1}{n}},
\end{equation*}
and $\text{Pr}(X_{i} \geq \max\left\{x,S_1,\dots,S_{\beta n},X_1,\dots, X_{i-1}\right\})$ is equal to
\begin{eqnarray*}
&&\text{Pr}(X_i \geq x | X_{i} \geq \max \left\{S_1,\dots,S_{\beta n},X_1,\dots, X_{i-1}
\right\})
\cdot
 \text{Pr}(X_{i} \geq \max \left\{S_1,\dots,S_{\beta n},X_1,\dots, X_{i-1} \right\}
)
\\
&=& 
\text{Pr}(\max\left\{S_1,\dots,S_{\beta n},X_1,\dots, X_{i} \right\}
\geq x)
\cdot
\text{Pr}(X_{i} = \max \left\{S_1,\dots,S_{\beta n},X_1,\dots,X_i\right\})
\\
&=&
\frac{1-F^{\beta n+i}(x)}{\beta n +i}.
\end{eqnarray*}
Thus, 
\begin{eqnarray*}
\text{Pr}(X_{\tau} \geq x)&=&
(\beta +\delta) \frac{1}{n} \sum_{i=\delta n}^n \frac{1-F^{\beta n+i}(x)}{(\beta +\frac{i}{n})(\beta+\frac{i-1}{n})} \geq (\beta +\delta) \frac{1}{n} \sum_{i=\delta n}^n \frac{1-F^{\beta n+i}(x)}{\left(\beta +\frac{i}{n} \right)^2}.
\end{eqnarray*}
Set $a=F^n(x)$, and notice that $\text{Pr}(\max\left\{X_1,\dots,X_n\right\} \geq x)=1-a^n$. Hence, to prove~\eqref{dom}, it is enough to prove that for all $a \in [0,1]$ and large $n$, 
\begin{equation} \label{dom2}
(\beta+\delta) \frac{1}{n} \sum_{i=\delta n}^n \frac{1-a^{\beta+\frac{i}{n}}}{(\beta +\frac{i}{n})^2} \geq \left(\frac{(1+\beta)}{e}+o(1) \right) (1-a),
\end{equation}
where the $o(1)$ is independent of $a$. Let $g(t)=\frac{1-a^{\beta+t}}{(\beta+t)^2}$. There exists $C>0$ such that for all $t \in [0,1]$ and $a \in [0,1]$, $|g'(t)| \leq C(1-a)$. By property of the Riemann integral, it follows that for all $a \in [0,1]$ and $n \geq 1$,
 \begin{eqnarray*}
\left| (\beta+\delta) \frac{1}{n} \sum_{i=\delta n}^n \frac{1-a^{\beta+\frac{i}{n}}}{(\beta +\frac{i}{n})^2} -(\beta+\delta) \int_{\delta}^1 \frac{1-a^{\beta+t}}{(\beta+t)^2} \;\mathrm{d}t  \right| \leq \frac{C(1-a)}{n}.
 \end{eqnarray*}
Since $\beta+\delta=\frac{1+\beta}{e}$, to prove~\eqref{dom2}, it is thus enough to prove that for all $a \in [0,1]$,
\begin{equation*}
\int_{\delta}^1 \frac{1-a^{\beta+t}}{(\beta+t)^2} \;\mathrm{d}t  \geq 1-a. 
\end{equation*}
The above inequality clearly holds for $a=1$, and thus by the change of variables $t'=t+\beta$, we want to prove that for all $a \in [0,1)$, 
\begin{equation*}
\int_{\frac{1+\beta}{e}}^{1+\beta} \frac{1-a^t}{(1-a)t^2} \;\mathrm{d}t \geq 1. 
\end{equation*}
It is enough to prove that the above integral is decreasing in $a$. Indeed, its limit as $a$ goes to 1 is $1$. Define 
\begin{equation*}
H(a)=\int_{\frac{1+\beta}{e}}^{1+\beta} \frac{1-a^t}{(1-a)t^2} \;\mathrm{d}t.
\end{equation*}
We have
\begin{equation*}
H'(a)=\int_{\frac{1+\beta}{e}}^{1+\beta} \frac{-(1-a)ta^{t-1}+(1-a^t)}{(1-a)^2t^2} \;\mathrm{d}t.
\end{equation*}
Thus, we want to prove that the function $I$ defined by
\begin{equation*}
I(a)=\int_{\frac{1+\beta}{e}}^{1+\beta} \frac{-(1-a)ta^{t-1}+(1-a^t)}{t^2} \;\mathrm{d}t
\end{equation*}
is negative. Notice that
\begin{equation*}
I(1)=0,
\end{equation*}
thus it is enough to prove that $I$ is increasing, which means that $I'$ is positive. We have 
\begin{equation*}
I'(a)=-\int_{\frac{1+\beta}{e}}^{1+\beta} \frac{a^{t-2}(1-a)(t-1)}{t} \;\mathrm{d}t.
\end{equation*}
Let
\begin{equation*}
J(a)=\int_{\frac{1+\beta}{e}}^{1+\beta} \frac{a^{t-1}(t-1)}{t} \;\mathrm{d}t.
\end{equation*}
Thus, we want to prove that $J$ is negative. For all $a \in [0,1)$ and $t \in [\frac{1+\beta}{e},1+\beta]$, we have
$\frac{a^{t-1}(t-1)}{t} \leq \frac{(t-1)}{t}$, and thus
\begin{equation*}
J(a) \leq \int_{\frac{1+\beta}{e}}^{1+\beta} \frac{t-1}{t} \;\mathrm{d}t=\left(1-\frac{1}{e}\right)(1+\beta)-1 \leq \left(1-\frac{1}{e}\right)\left(1+\frac{1}{e-1}\right)-1=0. 
\end{equation*}
This finishes the proof. 
\end{proof}

\section{Streaming Prophet Inequalities}

We will now argue that with arbitrarily small additive loss $\varepsilon$ in the guarantee, our MRS algorithms can be implemented as streaming algorithms. To this end, consider $x_0\in[0,1]$, $\bar{y}\in \mathbb{R}_+$, and the continuous function $g:[0,1]\rightarrow[0,\bar{y}]$ based on which the MRS algorithm is defined. The property that we need our functions $g$ to satisfy (and our algorithms from \autoref{sec:subset} do satisfy) is that the graph of $g$ has $O_\varepsilon(1)$ intersection points with the (infinitely many) horizontal lines at height $0,\varepsilon,2\varepsilon,\dots$. 

\begin{theorem}\label{thm:streaming}
	Let $\varepsilon>0$. Assume there exists an MRS algorithm with guarantee $\alpha$ for the unknown-distribution setting with $O(n)$ samples. Further assume that the MRS algorithm is based on continuous function $g$ with $|\{x\in[0,1]:\exists q\in\mathbb{N}: g(x)=q\cdot\varepsilon\}|=O_\varepsilon(1)$. Then there exists a streaming algorithm using $O_\varepsilon(\log n)$ space and achieving a guarantee of $\alpha-\varepsilon$ in the same setting.
\end{theorem}

The algorithm divides the $x$-range $[0,1]$ into strips of width $\varepsilon$ and the $y$-range $[0,\bar{y}]$ into strips of width $\varepsilon$. This creates $\gamma \leq \lceil 1/\varepsilon\rceil + O_\varepsilon(1)$ 
intersection points with $g$. Let $0=x_1 \leq x_2 \leq \dots \leq x_\gamma=1$ be the corresponding $x$-coordinates of these intersection points.

For all $i=1,\dots,\gamma-1$, the algorithm uses a single threshold that is distributed like the maximum of $\lceil g(x_i)\cdot n\rceil$ fresh samples for all steps in $[x_i\cdot n,x_{i+1}\cdot n)$. We observe that the emerging algorithm can be viewed as an MRS algorithm again. Towards this, let $\tilde{g}:[0,1]\rightarrow[0,\bar{y}+\varepsilon]$ be the function that is equal to $g(x_i)$ at $x_i$ and then grows linearly with slope $1$ until (and not including) $x_{i+1}$. We can (essentially) view the new algorithm as the MRS algorithm based on $\tilde{g}$.

\begin{lemma}\label{lem:streaming-gtilde}
	Let $j\in[x_i\cdot n,x_{i+1}\cdot n)\cap\mathbb{Z}$ for some $i\in\{1,\dots,\gamma-1\}$. Conditioned on arriving in step $j$, the above algorithm sets a threshold for $X_j$ that is distributed like the maximum of $\tilde{g}(j/n)\cdot n\pm O(1)$ fresh samples.
\end{lemma}
\begin{proof}
Let $j_0$ be the first integer in $[x_i\cdot n,x_{i+1}\cdot n)$. Denote the subset of values selected uniformly at random from $\{S_1,\dots,S_{\beta n}, X_1,\dots,X_{j_0-1}\}$ by $Y_S=\{S'_1,\dots,S'_\ell\}$ where $\ell=\lceil g(j_0/n)\cdot n\rceil=\tilde{g}(j_0/n)\cdot n\pm O(1)$ by continuity of $g$. By \autoref{lem:bayes}, this set and the set $Y=\{S'_1,\dots,S'_\ell,X_{j_0},\dots,X_{j-1}\}$ are distributed like sets of $\ell$ and $|Y|=\ell+(j-j_0)-1$, respectively, fresh samples. Note that $|Y|=\tilde{g}(j/n)\cdot n\pm O(1)$. It suffices to show that, conditioned on arriving in step $j_0$ and any such set $Y$, (i) the probability of arriving in step $j$ is independent of $Y$, and (ii), if the algorithm arrives in step $j$, the threshold it sets in step $j$ is $\max Y$.

Towards showing (i) and (ii), again condition on arriving in step $j_0$ and any set $Y$. Note that the algorithm arrives in step $j$ if and only if $\max Y=\max Y_S$. This implies that, throughout steps $j_0,\dots,j$, the algorithm sets $\max Y$ as threshold, showing (ii). Finally notice that, since both $Y_S$ and $Y\setminus Y_S$ are sets of fresh draws from $F$, $\max Y=\max Y_S$ happens with probability independent of $Y$, showing (i).	
\end{proof}

Further note that our construction ensures that $g(x) \leq \tilde{g}(x) \leq g(x)+ 2\varepsilon$ for all $x\in[0,1]$. See \autoref{fig:streaming-algo} for a visualization of the construction.

\begin{lemma}\label{lem:streaming-error-bound}
Suppose that the MRS algorithm defined by $g$ achieves approximation ratio $\alpha$ via~\eqref{eq:MRS-goal-pointwise} and that $|\tilde{g}(x) - g(x)| \leq 2\varepsilon$ for all $x\in[0,1]$, then the MRS algorithm defined by $\tilde{g}$ achieves approximation ratio $\alpha - O(\sqrt{\eps})$.
\end{lemma}
\begin{proof}
As shown previously, $g$ satisfies~\eqref{eq:MRS-goal-pointwise}:
For all $a\in[0,1]$,
\begin{equation}  \label{g_equation}
	\int_0^1\exp\left(-\int_0^y\frac{1}{g(z)}\;\mathrm{d}z\right)\cdot\frac{1}{g(y)}\cdot\left(1-a^{g(y)}\right)\mathrm{d}y\geq\alpha\cdot(1-a).
\end{equation}
Moreover, to prove our claim, it is enough to show that the above equation holds, replacing $g$ by $\tilde{g}$ and $\alpha$ by $\alpha-O(\sqrt{\eps})$. Note that, as $\varepsilon\rightarrow 0$ and uniformly in $a$, 
\begin{equation} \label{tilde_equation} 
	(1-a)^{-1}\int_{1-\varepsilon}^1 \exp\left(-\int_0^y\frac{1}{g(z)}\;\mathrm{d}z\right)\cdot\frac{1}{g(y)}\cdot\left(1-a^{g(y)}\right)\mathrm{d}y=O(\sqrt{\varepsilon}).
\end{equation}
 Second, for $y \in [0,1-\varepsilon]$, we have
$g(y) \geq \tilde{g}(y)-2 \varepsilon \geq \tilde{g}(y) \times (1-2\varepsilon/\tilde{g}(y))
\geq \tilde{g}(y)(1-2\varepsilon/g(1-\varepsilon))=\tilde{g}(y)(1-O(\sqrt{\varepsilon}))$, and thus $1/\tilde{g}(y) \geq 1/g(y)-O(\sqrt{\varepsilon})$. Hence, by~\eqref{g_equation} and~\eqref{tilde_equation}, as $\varepsilon$ tends to 0 and uniformly in $a$,
\begin{equation*}
	(1-a)^{-1}\int_0^{1-\varepsilon} \exp\left(-\int_0^y\frac{1}{\tilde{g}(z)}\;\mathrm{d}z\right)\cdot\frac{1}{\tilde{g}(y)}\cdot\left(1-a^{\tilde{g}(y)}\right)\mathrm{d}y \geq \alpha-O(\sqrt{\varepsilon}).
\end{equation*}
To obtain the above inequality, we have used in addition the fact that the left-hand side term and right-hand side term in the integrand of equation \eqref{g_equation} increase when one replaces $g$ by $\tilde{g}$. 
This completes the proof. 
\end{proof}

To implement our approach as a streaming algorithm, for each $i=0,\dots,\gamma$, we construct the maximum of the corresponding random subset on the fly: We count how many random positions are left to consider and include the current position with probability proportional to that count.

\begin{lemma}\label{lem:on-the-fly}
For each $x_i$ the threshold corresponding to $g(x_i)$ can be computed with a single pass over the data and $O(\log{n})$ space.
\end{lemma}
\begin{proof}
	Consider the first $j$ such that $j/n \geq x_i$ and let $q = \lceil g(x_i) \cdot n \rceil$. We will construct a $0/1$-vector of length $k+j-1$ with exactly $q$ many $1$'s on the fly such that the positions where the bit vector is $1$ correspond to a subset of size $q$ chosen uniformly at random without replacement from $\{S_1, \dots, S_k, X_1, \dots, X_{j-1}\}$. We can then compute the threshold in an online fashion by remembering the maximum $T$ of all values where we have set the bit to $1$.

We do this as follows: We remember the number $s$ of $1$'s that we still need and the number of positions $t$ still to come. Initially, $s = q$ and $t = k+j-1$. Then for $\ell = 1$ to $k+j-1$ we toss a biased coin that comes up $1$ with probability $s/t$ and is $0$ otherwise. If it comes up $1$ we update $s = s-1$ and $t = t-1$, otherwise we keep $s$ and just set $t = t-1$.

It now suffices to show that this process always yields a $0/1$-vector of length $k+j-1$ with exactly $q$ many $1$'s, and that all such bit vectors are equally likely. The former follows from the fact that the probability of seeing another $1$ is set to zero once there are already $q$ many $1$'s and that once the number of remaining positions equals the number of $1$'s that are still needed the probability of seeing a $1$ is set to one for all remaining steps.

It remains to show that all $0/1$-vectors of length $k+j-1$ with $q$ many $1$'s are equally likely, i.e., that the likelihood of seeing any such vector vector is $1/{k+j-1 \choose q}$. 
Indeed, consider an arbitrary such vector $z$. Let $E = \{e_1, e_2, \dots, e_q\} \subseteq [k+j-1]$ with $e_1 < e_2 < \dots < e_q$ be the indices $\ell$ where $z_\ell = 1$  and let $N = \{n_1, \dots, n_{k+j-1-q}\}$ with $n_1 < \dots < n_{k+j-1-q}$ be the indices $\ell$ where $z_\ell = 0$. Then,
\allowdisplaybreaks
\begin{align*}
\Pr{z} &= \prod_{\ell=1}^{q} \frac{q-\ell+1}{k+j-e_\ell} \cdot \prod_{\ell = 1}^{k+j-1-q} \frac{k+j-q-\ell}{k+j-n_\ell} 
= \frac{1}{{k+j-1 \choose q}}.
\end{align*}

The space complexity is $O(\log(n))$ because all the algorithm needs to store is the threshold, the remaining number of positions, and the number of ones that are still required.
\end{proof}

\autoref{thm:streaming} then follows by combining the above lemmata.

\begin{figure}
\begin{center}
\begin{tikzpicture}
\begin{axis}[axis x line = left, axis y line = left, x=8cm, y=1.7cm, xmin = 0.0, xmax = 1.1, ymin = 0.0, ymax = 1.6, font=\small, line width = 1pt]
\addplot [no marks] coordinates {
	( 0.0 , 1.4433756729740645 )
	( 0.01 , 1.4361406616345072 )
	( 0.02 , 1.4288690166235207 )
	( 0.03 , 1.4215601757693317 )
	( 0.04 , 1.4142135623730951 )
	( 0.05 , 1.4068285846778443 )
	( 0.06 , 1.399404635312222 )
	( 0.07 , 1.3919410907075054 )
	( 0.08 , 1.3844373104863459 )
	( 0.09 , 1.3768926368215255 )
	( 0.1 , 1.3693063937629153 )
	( 0.11 , 1.3616778865306827 )
	( 0.12 , 1.3540064007726602 )
	( 0.13 , 1.3462912017836262 )
	( 0.14 , 1.3385315336840842 )
	( 0.15 , 1.3307266185559428 )
	( 0.16 , 1.3228756555322954 )
	( 0.17 , 1.3149778198382918 )
	( 0.18 , 1.3070322617798436 )
	( 0.19 , 1.299038105676658 )
	( 0.2 , 1.2909944487358056 )
	( 0.21 , 1.282900359861721 )
	( 0.22 , 1.2747548783981963 )
	( 0.23 , 1.2665570127975554 )
	( 0.24 , 1.2583057392117916 )	
	( 0.25 , 1.25 )
	( 0.26 , 1.241638702145945 )
	( 0.27 , 1.233220715579062 )
	( 0.28 , 1.2247448713915892 )
	( 0.29 , 1.2162099599438687 )
	( 0.3 , 1.20761472884912 )
	( 0.31 , 1.19895788082818 )
	( 0.32 , 1.1902380714238086 )
	( 0.33 , 1.181453906563152 )
	( 0.34 , 1.1726039399558574 )
	( 0.35 , 1.1636866703140785 )
	( 0.36 , 1.1547005383792517 )
	( 0.37 , 1.14564392373896 )
	( 0.38 , 1.136515141415488 )
	( 0.39 , 1.1273124382057236 )
	( 0.4 , 1.118033988749895 )
	( 0.41 , 1.1086778913041726 )
	( 0.42 , 1.09924216318941 )
	( 0.43 , 1.0897247358851683 )
	( 0.44 , 1.0801234497346432 )
	( 0.45 , 1.0704360482220943 )
	( 0.46 , 1.0606601717798212 )
	( 0.47 , 1.0507933510765408 )
	( 0.48 , 1.0408329997330663 )
	( 0.49 , 1.0307764064044151 )
	( 0.5 , 1.0206207261596576 )
	( 0.51 , 1.0103629710818451 )
	( 0.52 , 1.0 )
	( 0.53 , 0.9895285072531597 )
	( 0.54 , 0.9789450103725609 )
	( 0.55 , 0.9682458365518541 )
	( 0.56 , 0.957427107756338 )
	( 0.57 , 0.9464847243000455 )
	( 0.58 , 0.9354143466934853 )
	( 0.59 , 0.9242113755341183 )
	( 0.6 , 0.9128709291752769 )
	( 0.61 , 0.9013878188659973 )
	( 0.62 , 0.8897565210026094 )
	( 0.63 , 0.8779711460710616 )
	( 0.64 , 0.8660254037844387 )
	( 0.65 , 0.8539125638299665 )
	( 0.66 , 0.841625411530173 )
	( 0.67 , 0.82915619758885 )
	( 0.68 , 0.816496580927726 )
	( 0.69 , 0.8036375634160795 )
	( 0.7 , 0.7905694150420948 )
	( 0.71 , 0.7772815877574013 )
	( 0.72 , 0.7637626158259735 )
	( 0.73 , 0.75 )
	( 0.74 , 0.7359800721939873 )
	( 0.75 , 0.7216878364870323 )
	( 0.76 , 0.7071067811865475 )
	( 0.77 , 0.6922186552431729 )
	( 0.78 , 0.67700320038633 )
	( 0.79 , 0.6614378277661477 )
	( 0.8 , 0.6454972243679027 )
	( 0.81 , 0.6291528696058956 )
	( 0.82 , 0.6123724356957945 )
	( 0.83 , 0.595119035711904 )
	( 0.84 , 0.577350269189626 )
	( 0.85 , 0.5590169943749475 )
	( 0.86 , 0.5400617248673216 )
	( 0.87 , 0.5204164998665333 )
	( 0.88 , 0.49999999999999994 )
	( 0.89 , 0.47871355387816916 )
	( 0.9 , 0.4564354645876384 )
	( 0.91 , 0.4330127018922192 )
	( 0.92 , 0.408248290463863 )
	( 0.93 , 0.38188130791298647 )
	( 0.94, 0.35355339059327373 )
	( 0.95 , 0.3227486121839512 )
	( 0.96 , 0.2886751345948128 )
	( 0.97 , 0.2500000000000002 )
	( 0.98 , 0.2041241452319315 )
	( 0.99 , 0.14433756729740685 )
	( 1.0 , 0.0 )
};
\foreach \x in {0.2,0.4,0.6,0.8,1.0}
	{
		\addplot [no marks, smooth, line width = 0.5pt,dotted] coordinates {
		(\x,0)
		(\x,1.6)
		};
	}
\foreach \y in {0.2,0.4,0.6,0.8,1.0,1.2,1.4,1.6}
	{
		\addplot [no marks, smooth, line width = 0.5pt,dotted] coordinates {
		(0,\y)
		(1,\y)
		};
	}
\addplot  [no marks, smooth, line width = 1pt, color=green!50!black] coordinates {
	(0,1.44338)
	(0.0592,1.50258)
};
\addplot  [no marks, smooth, line width = 1pt, color=green!50!black] coordinates {
	(0.0592,1.4)
	(0.2,1.5408)
};
\addplot  [no marks, smooth, line width = 1pt, color=green!50!black] coordinates {
	(0.2, 1.29099)
	(0.3088,1.39979)
};
\addplot  [no marks, smooth, line width = 1pt, color=green!50!black] coordinates {
	(0.3088, 1.2)
	(0.4, 1.2 + 0.4 - 0.3088)
};
\addplot  [no marks, smooth, line width = 1pt, color=green!50!black] coordinates {
	(0.4, 1.11803)
	(0.52, 1.11803 + 0.52- 0.4)
};
\addplot  [no marks, smooth, line width = 1pt, color=green!50!black] coordinates {
	(0.52, 1)
	(0.6, 1 + 0.6- 0.52)
};
\addplot  [no marks, smooth, line width = 1pt, color=green!50!black] coordinates {
	(0.6, 0.912871)
	(0.6928, 0.912871 + 0.6928 - 0.6)
};
\addplot  [no marks, smooth, line width = 1pt, color=green!50!black] coordinates {
	(0.6928, 0.8)
	(0.8, 0.8 + 0.8 - 0.6928)
};
\addplot  [no marks, smooth, line width = 1pt, color=green!50!black] coordinates {
	(0.8, 0.645497)
	(0.8272, 0.645497 + 0.8272 - 0.8)
};
\addplot  [no marks, smooth, line width = 1pt, color=green!50!black] coordinates {
	(0.8272, 0.6)
	(0.9232, 0.6 + 0.9232 - 0.8272)
};
\addplot  [no marks, smooth, line width = 1pt, color=green!50!black] coordinates {
	(0.9232, 0.4)
	(0.9808, 0.4 + 0.9808 - 0.9232)
};
\addplot  [no marks, smooth, line width = 1pt, color=green!50!black] coordinates {
	(0.9808, 0.2)
	(1, 0.2 + 1 - 0.9808)
};
\end{axis}
\end{tikzpicture}
\end{center}
\caption{Visualization of the $g$ function from \autoref{subsec:unconstrained} (bold, black), the grid with $\varepsilon = 0.2$ (dotted, black), and the resulting $\tilde{g}$ function (green, bold).}
\label{fig:streaming-algo}
\end{figure}
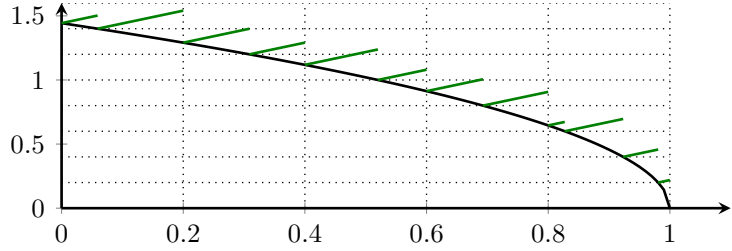

\section{Unknown Distribution With Prior}  \label{sec:exchangeable}
 
Let $\theta$ be a distribution over a set of real distributions $\left\{F^1,\dots,F^m\right\}$, where $m \geq 1$. 
Consider the following learning problem, which we will call \emph{unknown i.i.d.\@ problem with prior $\theta$}:
\begin{itemize}
\item Nature draws $F$ according to $\theta$ and $(X_1,\dots,X_n)$ i.i.d.\@ according to $F$;
\item the gambler knows $\theta$ but not $F$, and has to choose a stopping rule with stopping time $\tau$ in order to maximize $\mathbb{E}(X_\tau)$. 
\end{itemize}
The key difference with the unknown i.i.d.\@ problem without samples studied by \citet{CorreaDFS19} is that the gambler \emph{knows} $\theta$. Thus, it is plausible that a constant better than $1/e$ can be guaranteed in this model. 
Our main result in this section shows that $1/e$ remains best possible.
\begin{theorem} \label{thm:known}
For any $\delta>0$, there exists $n_0 \in \mathbb{N}$ such that for any $n \geq n_0$, there exists a finite number of finitely supported distributions $F^1,\dots,F^m$ and a distribution $\theta$ over $\left\{F^1,\dots,F^m\right\}$ such that for any 
stopping rule with stopping time~$\tau$,
\[
\mathbb{E}[X_\tau] \leq \left(\frac{1}{e}+\delta\right) \cdot \mathbb{E}[\max\{X_1,\dots,X_n\}].
\]
\end{theorem}
An important consequence of \autoref{thm:known} is that it answers a question of \citet{HillKertz92}. Indeed, in their survey on the state of the art in prophet inequalities, \citeauthor{HillKertz92} stated five groups of open questions. The second group, Question~2 on Page~203, concerns exchangeable random variables. Regarding the maximum value obtainable by a stopping rule~$\tau$ and the expected maximum value in the sequence they specifically asked for the largest universal constant $\gamma$ such that,
if $X_1,X_2,\dots,X_n$ are exchangeable random variables taking values in $[0,1]$ then
\begin{equation} \label{constant}
	\mathbb{E}[X_\tau] \geq \gamma \cdot \mathbb{E}\left[\max\{X_1,\dots,X_n\}\right].
\end{equation}
We resolve this question by showing that the universal constant is $\gamma=1/e\approx 0.368$. Prior to this work it was only known that the tight approximation ratio for $n=2$ is $1/2$,  that the ratio is weakly decreasing in~$n$, and that it is at least $1/e$~\citep{EltonKertz91}.

Note that the assumption that $X_1,\dots,X_n$ take value in $[0,1]$ is actually without loss of generality. Indeed, given non-negative random variables $X_1,\dots,X_n$ and $M\geq 1$, 
one can define $X_i'=(X_i/M) 1_{X_i \leq M}$ for all $i\in[n]$. The variables $X'_1,\dots,X'_n$ then are exchangeable and take values in $[0,1]$, and for $M$ large enough, the best approximation guarantees achievable for both instances are almost the same. We will therefore omit the assumption when stating our results.

\medskip

\begin{corollary}
\label{thm:exch}
For any $\delta>0$, there exists $n_0 \in \mathbb{N}$ such that for any $n \geq n_0$, there exists a sequence of exchangeable random variables $X_1,X_2,\dots,X_n$ such that for any stopping rule with stopping time~$\tau$,
\[
\mathbb{E}[X_\tau] \leq \left(\frac{1}{e}+\delta\right) \cdot \mathbb{E}[\max\{X_1,\dots,X_n\}].
\]
In particular, the largest constant $\gamma$ in~\eqref{constant} is $1/e$. 
\end{corollary}
\begin{proof}
Let $\delta>0$, $n_0 \in \mathbb{N}$, $n \geq n_0$, $m \geq 1$, $F^1,\dots,F^m$ and $\theta$ be given by \autoref{thm:known}. Let $F$ be a random variable distributed according to $\theta$, and
$(X_1,\dots,X_n)$ a sequence of random variables such that its distribution conditional to $F$ is the one of i.i.d.\@ random variables from $F$. The random variables $(X_1,\dots,X_n)$ are exchangeable, 
and by \autoref{thm:known}, no stopping rule can guarantee a constant better than $1/e+\delta$. 
\end{proof}
The proof of \autoref{thm:known} extends on a technique developed recently by \citet{CorreaDFS19} to show a lower bound of~$1/e$ for unknown i.i.d.\@ random variables with finite support and without any additional samples. We modify the central element of this technique, which relies on the infinite version of Ramsey's theorem, to instead use the finite version (\autoref{sec:unknown}). We then deduce the desired impossibility using a minimax argument (\autoref{sec:minmax}).

\subsection{Hard Finite Instances for Unknown Distribution} \label{sec:unknown}

Let us return to the setting considered by \citet{CorreaDFS19}, where random variables are drawn independently from the same unknown distribution and we do not have access to any additional samples and there is no prior over the unknown distribution. For $p\in\N$, denote by $B_p$ the set of probability measures on $[p]$. 
We prove the following result.
\begin{proposition} \label{prop:supp}
For all $\delta>0$, there exists $n_0 \in \N$ such that for any $n \geq n_0$, there exists $p \in \N$ such that for any stopping rule $\vr$ with associated stopping time $\tau$, there exists $b \in B_p$ such that 
when $X_1,\dots,X_n$ are i.i.d.\@ random variables drawn from $b$, 
\[
\mathbb{E}[X_\tau] < \left(\frac{1}{e}+\delta \right) \cdot \mathbb{E}[\max\{X_1,\dots,X_n\}].
\]
\end{proposition}

What distinguishes this result from Theorem~3.2 of \citeauthor{CorreaDFS19} is that the distribution~$b$ has finite support and the cardinality of its support is independent of the stopping rule. This property will be crucial when applying the minimax theorem in \autoref{sec:minmax}. 

The result can be obtained by modifying the central construction in the proof of \citeauthor{CorreaDFS19} while keeping much of the structure of that proof intact. To make it easier to compare the two results we will follow the original structure, and begin by recalling the definition of oblivious stopping rules.
\begin{definition}\label{def:val-obl}
	Let $\eps>0$ and $V\subset\N$. 
	\begin{itemize}
	\item
	A stopping rule $\vr$ is \emph{$(\eps,i)$-value-oblivious on $V$} if, there exists a $q_i\in[0,1]$ such that, for all pairwise distinct $v_1,\dots,v_i\in V$ with $v_i>\max\{v_1,\dots,v_{i-1}\}$, it holds that $r_i(v_1,\dots,v_i)\in[q_i-\eps,q_i+\eps)$.	
	\item
	A stopping rule $\vr$ is \emph{$\eps$-value-oblivious on $V$} if, for all $i\in[n]$, it is 
	$(\eps,i)$-value-oblivious on $V$. 
	\item
	A stopping rule $\vr$ is \emph{order-oblivious} if for all $j\in[n]$, all pairwise distinct $v_1,\dots,v_j\in\mathbb{R}_+$ and all permutations $\pi$ of $[j-1]$, $r_i(v_1,\dots,v_j)=r_i(v_{\pi(1)},\dots,v_{\pi(j-1)},v_j)$.
	\end{itemize}
\end{definition}

The cornerstone of our proof is the following lemma.
\begin{lemma}\label{lem:structure1}
	Let $\eps>0$. For any $n \in \N$, there exists $p \in \N$ such that if there exists a stopping rule with guarantee $\alpha$, then there exists a stopping rule~$\vr$ with guarantee $\alpha$ such that $\vr$ is $\eps$-value-oblivious on $V$, for some finite set $V \subset [p]$ with cardinality $n^3+1$. 
\end{lemma}
The difference to Lemma~3.4 of \citeauthor{CorreaDFS19} is that the set $V$ is finite, and in addition is uniformly bounded by an integer $p$ that depends only on~$n$. Consequently, instead of the infinite version of Ramsey's theorem used by \citeauthor{CorreaDFS19}, we need the following finite version given for example by \citet{CFS10}.
\begin{lemma}
\label{lem:ramsey}
There exists a function $R:\N^3 \rightarrow \N$ such that for all $n \geq 1$, for all complete $m$-hypergraph with $c$ colors and order larger than $R(m,n,c)$, there exists a sub-hypergraph of order $n$ that is monochromatic.  
\end{lemma}
\begin{proof}[Proof of \autoref{lem:structure1}]
Fix $\eps>0$ and set $c=\lfloor (2 \eps)^{-1} \rfloor$.
Define an integer sequence $(p_i)_{0 \leq i \leq n}$ by induction as $p_n=n^3+1$ and 
$p_{i-1}=R(i,p_i,c)$.
Consider a stopping rule $\vr$ with guarantee $\alpha$. By Lemma~3.6 of \citeauthor{CorreaDFS19} it is without loss of generality to assume that $\vr$ is order-oblivious. We show by induction on $j \in \left\{0,1,\dots,n\right\}$ that there exists a set $S^j\subset S^0$ such that 
$|S^j|=p_j$ and for all $i\in\left[j\right]$, $\vr$ is $(\eps,i)$-value-oblivious on $S^j$.

The set $S_0=\left\{n^{3s}:s=0,\dots,p_0 \right\}$ satisfies the induction hypothesis for $j=0$. We proceed to show it for $j>0$. First, observe that we only need to find a set $S^j \subset S^{j-1}$ such that $|S_j|=p_j$ and $\vr$ is $(\eps,j)$-value oblivious on $S^j$, because it follows from the induction hypothesis that for all $i \in [j-1]$, $\vr$ is $(\eps,i)$-value-oblivious on $S^i$ and thus on the subset $S^j \subset S^i$. 

Toward the application of \autoref{lem:ramsey}, we construct a complete $j$-hypergraph $H$ with vertex set $S^{j-1}$. Consider any set ${v_1,...,v_j} \subset S^{j-1}$ of cardinality $j$ such that 
$v_j > max(v_1,...,v_{j-1})$. Note that there exists a unique $u\in \left\{1,2,...,c\right\}$ such that $r_j(v_1,...,v_j) \in [(2u-1) \eps-\eps,(2u-1) \eps+\eps)$, and color the hyperedge $\left\{v_1,...,v_j\right\}$ of $H$ with color $u$. By \autoref{lem:ramsey}, there exists a finite set $S^j$ of vertices with cardinality $p_j$ that induces a complete monochromatic sub-hypergraph of~$H$. Let~$u$ be the color of this sub-hypergraph, set $q=(2u-1)\eps$, and consider distinct $v_1,...,v_j \in S^j$ with $v_j>max(v_1,...,v_{j-1})$. Since the edge $\left\{v_1,...,v_j\right\}$ in $H$ has color $u$, $r_j(v_{\pi(1)},...,v_{\pi(j-1)},v_j) \in [q-\eps,q+\eps)$ for some permutation $\pi$ of $S^{j-1}$. But since $\vr$ is order-oblivious, also $r_j(v_1,...,v_{j-1},v_j) \in [q-\eps,q+\eps)$. So $\vr$ is $(\eps,j)$-value oblivious on $S^j$. This completes the induction step.
\end{proof}

With \autoref{lem:structure1} at hand we are now ready to prove \autoref{prop:supp}.

\begin{proof}[Proof of \autoref{prop:supp}]
Let $\delta>0$ and $n \in \N$. We proceed by contradiction, and consider a stopping rule $\vr$ with performance guarantee $1/e+\delta$. Set $\eps=1/n^2$. By \autoref{lem:structure1}, there exists a stopping rule $\vr$ with performance guarantee~$1/e+\delta$ and a set $V\subset [p]$ with cardinality $n^3+1$ on which~$\vr$ is $\eps$-value-oblivious. Let $u$ be the maximum of $V$, and write $V=\left\{v_1,\dots,v_{n^3},u\right\}$. Denote by $\tau$ the stopping time of~$\vr$. 
By construction, we have $u\geq n^3\max\{v_1,\dots,v_{n^3}\}$. The rest of the proof proceeds in the same way as the proof of Theorem~3.2 of \citeauthor{CorreaDFS19}, and we give an informal summary for completeness. For each $i\in[n]$, let 
\[
	X_i = \begin{cases}
		v_1 & \text{w.p.\@ $\frac{1}{n^3}\cdot(1-\frac{1}{n^2})$} \\
		\vdots & \\
		v_{n^3} &\text{w.p.\@ $\frac{1}{n^3}\cdot(1-\frac{1}{n^2})$} \\
		u &\text{w.p.\@ $\frac{1}{n^2}$}
	\end{cases}.
\]
For this particular instance, the performance of any stopping rule corresponds approximately to the probability of picking $u$. Let us, therefore, investigate the probability that $\vr$ picks $u$. First note that with probability almost one, $X_1,\dots,X_n$ are distinct. Moreover, because $\vr$ is $\eps$-value-oblivious on~$V$, it can be changed with an error of $\eps=n^{-2}$ into a stopping rule that considers only the relative ranks of the values it has seen before making its decision. As there are only~$n$ stages, the error is insignificant.

The problem thus reduces to the classic secretary problem, for which it is known that no stopping rule can guarantee a probability of picking the maximum that is higher than $1/e+o(1)$ as $n$ goes to infinity~\citep{Ferguson89}. However, the stopping rule constructed from $\vr$ considers only relative ranks and selects the maximum with probability at least $1/e+\delta-o(1)$, which is a contradiction.
\end{proof}

\subsection{From Unknown to Unknown with Prior: A Minimax Argument}  \label{sec:minmax}

We now use a minimax argument to convert \autoref{prop:supp}, which concerns the unknown i.i.d.\@ case, into an upper bound (impossibility result) for the case of \textit{unknown i.i.d.\@ with prior}. Fix $\delta>0$ and let $n_0$ and $p$ be as in \autoref{prop:supp}. Fix $n\geq n_0$. We need the following definitions.

\newpage

\begin{definition} \mbox{}
\begin{itemize}
	\item A deterministic stopping rule is a sequence $a=(a_1,...,a_n)$ such that $a_i:[p]^{i-1} \rightarrow \left\{0,1\right\}$. Denote the set of such rules by~$A$. 
	\item A mixed stopping rule is a distribution over $A$. Denote the set of such rules by $\mathcal{P}(A)$.
	\item A behavior stopping rule is a sequence $\vr=(r_1,...,r_n)$ such that $r_i : [p]^{i-1} \rightarrow [0,1]$. Denote the set of such rules by~$C$. 
\end{itemize}
\end{definition}
Note that we have defined each class of stopping rule to consider only values in $[p]$ as inputs because we will consider random variables supported on~$[p]$. As before, considering inputs in $\R$ would not have a significant effect on the performance guarantee. Note further that behavior stopping rules correspond to stopping rules as defined in \autoref{sec:prelims}. The minimax argument will require us to consider mixed stopping rules. We will apply Kuhn's theorem~\citep{Kuhn53} to prove that mixed stopping rules and behavior stopping rules provide the same performance guarantee. 

According to our purpose any of these stopping rules will have the two interpretations, (i)~as a stopping rule in the unknown i.i.d.\@ problem, and (ii)~as a stopping rule in the i.i.d.\@ problem with some prior distribution 
$\theta$, where in the latter case we have omitted the dependence of $a$ on $\theta$. 

Recall that $B_p$ is the set of probability distributions over $[p]$. 
Because $p$ has been fixed and for ease of exposition, we will henceforth write $B$ instead of $B_p$. Let $\mathcal{P}(B)$ be the set of probability distributions over $B$. For given $a\in A$ and $b\in B$, define
\begin{equation*}
	g(a,b) = \E{X_{\tau}}-\left(\frac{1}{e}+\delta\right)\cdot\E{\max\{X_1,\dots,X_n\}},
\end{equation*}
where $\tau$ is the stopping time of the stopping rule $a$ in the i.i.d.\@ problem where $X_1,X_2,...,X_n$ are drawn from $b$. We can extend $g$ linearly to $\mathcal{P}(A)\times\mathcal{P}(B)$ by letting 
\begin{equation*}
	g(x,y) = \int_A \int_B g(a,b)  x(\mathrm{d}a) y(\mathrm{d}b).
\end{equation*}
Let
\begin{equation*}
V^-=\max_{x \in \mathcal{P}(A)} \min_{b \in B} g(x,b) \quad \text{and} \quad
V^+=\min_{y \in \mathcal{P}(B)} \max_{a \in A} g(a,y) .
\end{equation*}
Note that in the above expressions, by linearity of $g$ with respect to $x$ and $y$, 
$\min_{b \in B} g(x,b)=\min_{y \in \mathcal{P}(B)} g(x,y)$ and $\max_{a \in A} g(a,y)=\max_{x \in \mathcal{P}(A)} g(x,y)$. 
The key point is that~$V^-$ is related to the universal constant in the unknown i.i.d.\@ problem, while~$V^+$ is related to the universal constant in the unknown i.i.d.\@ problem with prior. 
Indeed, the following proposition holds:
\begin{proposition}
If $V^+<0$, then there exists $y \in \mathcal{P}(B)$ a finitely supported distribution over $B$ such that, in the unknown i.i.d. problem with prior $y$, for all behavior stopping rule with stopping time $\tau$,  
\[
\mathbb{E}[X_\tau] \leq \left(\frac{1}{e}+\delta\right) \cdot \mathbb{E}[\max\{X_1,\dots,X_n\}].
\]
\end{proposition} 
\begin{proof}
	Assume that $V^+<0$. By \citep[Proposition I.1.9]{MSZ}, there exists $y\in\mathcal{P}(B)$ with finite support such that $\max_{a \in A} g(a,y) \leq 0$. This means that no deterministic stopping rule provides a better guarantee than $(1/e+\delta)$ in the unknown i.i.d.\@ problem with prior distribution~$y$. Then no distribution over deterministic stopping rule can provide a better guarantee by linearity of expectation, neither can a behavior stopping rule by Kuhn's theorem \citep{Kuhn53}, which proves the result.
\end{proof}
To complete the proof of \autoref{thm:known} it is thus enough to show that $V^+< 0$. To this aim, we first prove that $V^+=V^-$, and then that $V^-<0$. 
\begin{proposition}
We have $V^+=V^-$.
\end{proposition}
\begin{proof}
The set $A$ is finite, the set $B$ compact metric. Moreover, for all $a\in A$, the mapping $b\rightarrow g(a,b)$ is continuous. By the minimax theorem~\citep[Proposition I.1.9]{MSZ} it follows that the mixed extension of the normal-form game $(A,B,g)$ has a value, \ie that $V^+=V^-$.
\end{proof}

To conclude the proof of \autoref{thm:known} it remains to show that $V^-< 0$. First note that by Kuhn's theorem \citep{Kuhn53},
\begin{equation*}
	V^-=\max_{r \in C} \min_{b \in B} \E{X_{\tau}}-\left(\frac{1}{e}+\delta\right)\cdot\E{\max\{X_1,\dots,X_n\}}.
\end{equation*}
By definition of $p$, which we have chosen as in \autoref{prop:supp}, for each stopping rule in~$C$ there exists $b\in B$ such that
\begin{equation*}
	\E{X_{\tau}}-\left(\frac{1}{e}+\delta\right)\cdot\E{\max\{X_1,\dots,X_n\}} < 0. 
\end{equation*}
It follows that $V^-<0$, which proves \autoref{thm:known}.

\bibliographystyle{abbrvnat}
\bibliography{abb,bibliography}

\end{document}